\documentclass[twoside,10pt]{article}
\pdfoutput=1

\usepackage{graphicx}

\usepackage{amsmath,amssymb,stmaryrd,mathrsfs}
\usepackage{cite} 
\usepackage[usenames,dvipsnames]{xcolor} 	
\definecolor{linkcolor}{rgb}{0,0,0} 
\usepackage[USenglish]{babel}
\usepackage{etex}
\usepackage[ruled,linesnumbered,algosection,vlined]{algorithm2e} 

\usepackage{caption}
\usepackage{amsthm}
\usepackage{enumerate}
\usepackage{makeidx}
\usepackage{array}
\usepackage{multirow}	  
\usepackage{hhline}	  
\usepackage{mathtools}    
\usepackage{mathdots}     
\usepackage[all,2cell]{xy} \UseAllTwocells \SelectTips{eu}{10} 
\usepackage[top=3.4cm,bottom=3.5cm,left=2.5cm,right=2.5cm]{geometry}

\theoremstyle{plain}
\newtheorem{theorem}{Theorem}
\newtheorem{pardgm}[theorem]{Paradigm}
\newtheorem{remark}[theorem]{Remark}
\newtheorem{proposition}[theorem]{Proposition}
\newtheorem{corollary}[theorem]{Corollary}
\newtheorem{lemma}[theorem]{Lemma}

\theoremstyle{definition}
\newtheorem{definition}[theorem]{Definition}

\newcommand{\set}[1]{\left\lbrace\, #1 \,\right\rbrace}

\newcommand{\IR}{{\mathbb{R}}}

\newcommand{\abst}[1]{\, #1 \,} 
\newcommand{\mytabvspace}{\vphantom{${X^X}^X$}} 

\newcommand{\Comment}[1]{\tcc*[r]{#1}}   
\newcommand{\IfComment}[1]{\tcc*[f]{#1}} 
\newcommand{\algoand}{\textbf{ and }} 
\newcommand{\pforest}{\texttt{p4est}\xspace} 
\newcommand{\tetcode}{\texttt{t8code}\xspace} 
\newcommand{\adef}{\, {:=} \,}
\newcommand{\agl}{\abst{=}}
\newcommand{\indexset}[1]{\set{0,\dots,#1-1}} 
\newcommand{\ohne}{\backslash}

\graphicspath{{./pictures}}

\title{Coarse mesh partitioning for tree based AMR}

\author{%
Carsten Burstedde\footnotemark[1]
\and
Johannes Holke\footnotemark[1]
}

\begin{document}

\maketitle

\thispagestyle{plain}

\renewcommand{\thefootnote}{\fnsymbol{footnote}}
 \footnotetext[1]{Institut f\"ur Numerische Simulation (INS) and
                 Hausdorff Center for Mathematics (HCM),\\%
                 Rheinische Friedrich-Wilhelms-Universit\"at Bonn, Germany}
\renewcommand{\thefootnote}{\arabic{footnote}}

\begin{abstract}
  In tree based adaptive mesh refinement, elements are partitioned between
  processes using a space filling curve.
  The curve establishes an ordering between all elements that derive from the
  same root element, the tree.
  When representing more complex geometries by patching together several trees,
  the roots of these trees form an unstructured coarse mesh.
  We present an algorithm to partition the elements of the coarse mesh such
  that (a) the fine mesh can be load-balanced to equal element counts per
  process regardless of the element-to-tree map and (b) each process that holds
  fine mesh elements has access to the meta data of all relevant trees.
  As an additional feature, the algorithm partitions the meta data of relevant
  ghost (halo) trees as well.
  We develop in detail how each process computes the communication pattern for
  the partition routine without handshaking and with minimal data movement.
  We demonstrate the scalability of this approach on up to 917e3 MPI ranks and
  .37e12 coarse mesh elements, measuring run times of one second or less.
\end{abstract}



\newcommand{\sZ}{\mathbb{Z}}
\newcommand{\ie}{i.e.\xspace}

\section{Introduction}

Adaptive mesh refinement (AMR) has a long tradition in numerical mathematics.
Over the years, the technique has evolved from a pure concept to a practical
method, where the transition may be located somewhere in the mid-80s; see, for
example, \cite{BergerOliger84}.
A second transition from serial to parallel computing, eventually parallelizing
the storage of the mesh itself, has occured around the turn of the millennium
(see e.g.\ \cite{GriebelZumbusch98}).
Different flavors of the approach have been studied, varying in the shape of
the mesh elements used (triangles, tetrahedra, quadrilaterals, cubes) and in
the logical grouping of the elements.

Elements may actually not be grouped at all and assembled into an unstructured
mesh, some recent references being \cite{NortonLyzengaParkerEtAl04,
ZhouSahniDevineEtAl10, RasquinSmithChitaleEtAl14}.
The connectivity of elements can be modeled as a graph, and the partitioning of
elements between parallel processes can be translated into partitioning the
graph.
Finding an approximate solution to this problem has been the subject of
extensive research, some of which has lead to the development of software
libraries \cite{KarypisKumar98, DevineBomanHeaphyEtAl02, ChevalierPellegrini08}.
In practice, the graph approach is often augmented by diffusive migration of
elements.
All in all, partitioning times of roughly 1e3 to 1e4 elements per second have
been measured \cite{CatalyurekBomanDevineEtAl07, SmithRasquinIbanezEtAl15,
CoupezSilvaDigonnet16}.
We may thus think of approaching the partitioning problem differently, trading
the generality of the graph for a mathematical structure that offers rates of
maybe 1e5 to 1e6 elements per second.

When we group elements by their size $h$, a particularly strict ansatz is the
hierarchical, $h = 2^{-\ell}$, using some refinement level $\ell \in \sZ$.
For square/cubic elements, we may introduce a rectangular grid of equal-sized
ones and reduce the assembly of the mesh to the question of arranging such
grids relative to each other.
Block-structured methods do just that in varying instances of generality, some
allowing free rotation and overlapping of blocks of any level
\cite{sk-ol-st:1989}, others imposing strict rules of assembly (such as not
allowing rotation, but only shifts in multiples of $h$ \cite{be-co:1989}).
Another such rule interprets elements as nodes of a tree, where the level
$\ell$ takes on a second meaning as the depth of a node below the root
\cite{RheinboldtMesztenyi80}.

Tree-based AMR can be implemented for any shape of element, where special 2D
solutions exist \cite{MeisterRahnemaBader16} as well as the popular
rectangles (2D) and cubes (3D) approach.
The tree root is identified with the largest possible element and thus inherits
its shape as a volume in space.
This points directly at a practical limitation: how to represent domain
geometries that have more complex shapes?
One answer is to consider multiple tree roots and arrange them in the fashion
of unstructured AMR, giving rise to a forest of elements
\cite{SteinmanMilnerNorleyEtAl03, StewartEdwards04, BangerthHartmannKanschat07}.
The mesh of trees is sometimes called the coarse mesh, which is created a
priori to map the topology and geometry of the domain with sufficient fidelity.
Elements may then be refined and coarsened recursively, changing the mesh below
the root of each tree.

Many simulations require one tree only, and the concept of the forest is not
called upon.
Popular forest AMR codes respect this fact by operating with no or minimal
overhead in the special case of a one-tree forest.
When multiple trees are needed, their number is limited by the available memory
of contemporary computers to roughly a million per process
\cite{BursteddeWilcoxGhattas11}.
While this number allows to execute most coarse meshing tasks in serial, we may
still ask how to work with coarse meshes of say a billion or more trees total,
such numbers being common in industrial and medical unstructured meshing.
In this case, the coarse mesh needs to be partitioned between the processes,
either by means of file I/O \cite{Bryan16} or in-core, as we will discuss
in this paper.

Most tree-based AMR codes make use of a space filling curve to order the
elements within a tree \cite{GriebelZumbusch02, TuOHallaronGhattas05,
SundarSampathBiros08} as well as points or other primitives
\cite{RahimianLashukVeerapaneniEtAl10}.
Two main approaches for partitioning a forest of elements have been discussed
\cite{Zumbusch03}, namely
(a) assigning each tree and thus all of its elements to one owner process
\cite{SelwoodBerzins99, BurriDednerKloefkornEtAl06}
or
(b) allowing a tree to contain elements belonging to multiple processes
\cite{BangerthHartmannKanschat07, BursteddeWilcoxGhattas11}.
The first approach offers a simpler logic, but may not provide acceptable
load balance when the number of elements differs vastly between trees.
The second allows for perfect partitioning of elements by number (the local
numbers of elements between processes differ by at most one), but presents the
issue of trees that are shared between multiple processes.

We choose paradigm (b) for speed and scalability, challenging us to solve an
$n$-to-$m$ communication problem for every coarse mesh element.
The objective of this paper is thus to develop how to do this without
handshaking (\ie, without having to determine separately
which
process receives from which), and with a minimal number of senders, receivers,
and messages.
One corner stone is to avoid identifying a single owner process for each tree,
but rather to treat all
its sharer processes as algorithmically active under the premise that they
produce a disjoint union of the information necessary to be transferred.
In particular, each process shall store the relevant tree meta data to be
readily available, eliminating the need to transfer this data from a single
owner process.

In this paper, we also integrate the parallel transfer of ghost trees.
The reason is that each process will eventually collect ghost elements, \ie,
remote elements adjacent to its own.
Although we do not discuss such an algorithm here, we note that ghost
elements of any process may be part of trees that are not in its local set.
To disconnect the ghost algorithm from identifying and transfering ghost trees,
we perform this as part of the coarse mesh partitioning, presently across tree
faces.
We study in detail what information we must maintain to reference neighbor
trees of ghost trees (that may themselves be either local, ghost, or neither)
and propose an algorithm with minimal communication effort.

We have implemented the coarse mesh partitioning for triangles and tetrahedra
using the SFC designed in \cite{BursteddeHolke16}, and for quadrilaterals and
cubes exploiting the logic from \cite{BursteddeWilcoxGhattas11}.
To demonstrate that our algorithms are safe to use, we verify that (a) small
numbers of trees require run times on the order of milliseconds and thus
present no noticeable overhead compared to a serial coarse mesh, and that (b)
the coarse mesh partitioning adds only a fraction of run time compared to the
partitioning of the forest elements, even for extraordinarily large numbers of
trees.
We show a practical example of 3D dynamic AMR on 8e3 cores using 383e6 trees
and up to 25e9 elements.
To investigate the ultimate limit of our algorithms, we partition coarse meshes
of up to .37e12 trees on a Blue Gene/Q system using 458e3 cores, obtaining a
total run time of about 1s, a rate of 7e5 trees per second.

We may summarize our results by saying that partitioning the trees can be made
no less costly than partitioning the elements, and often executes so fast that
it does not make a difference at all.
This allows a forest code that partitions both trees and elements dynamically
to treat the whole continuum of forest mesh
scenarios, from one tree with nearly trillions of elements on the one extreme
to billions of trees that are not refined at all on the other, with comparable
efficiency.

\section{Forest based AMR}

We understand the connectivity of a forest as a mesh of root elements, which is
effectively the coarsest possible mesh.
We will simply call it \emph{coarse mesh} in the following.
Each root element may be subdivided recursively into elements, replacing one
element by four (triangles and quadrilaterals) or eight (tetrahedra and cubes).
For our purposes, the subdivision may be chosen freely by an application.
Thus, in our context an element has two roles: one geometric as a volume in
space and one logical as node of a tree.
The leaf elements of the forest compose the \emph{forest mesh} used for
computation, see also Figures~\ref{fig:hybridmesh} and \ref{fig:twotreeforest}.
The (leaf) elements may be used in a classical finite element/finite volume/dG
setting or as a meta structure, for example to manage a subgrid in each leaf
element \cite{DreherGrauer05, SchiveTsaiChiueh10, BursteddeCalhounMandliEtAl14,
MaxwellKolletSmithEtAl16}.

The order of elements is established first by tree and then by their index
with respect to a space filling curve (SFC, see also \cite{Bader12}):

We enumerate the $K$ trees of the coarse mesh by $0,\dots,{K-1}$ and
call the number $k$ of a tree its \emph{global index}.
With the global index we naturally extend the SFC order of the leaves:
Let a leaf element of the tree k have SFC index $I$ (within that tree), then we
define the combined index $(k, I)$.
This index compares to a second index $(k', J)$ as
\begin{equation}
  (k, I) < (k', J) \quad :\Leftrightarrow \quad
  \text{$k < k'$ or ($k = k'$ and $I < J$)}
  .
\end{equation}
In practice we store the mesh elements local to a process in one contiguous
array per locally non-empty tree in precisely this order.

\begin{figure}
\center
\includegraphics[width=0.45\textwidth]{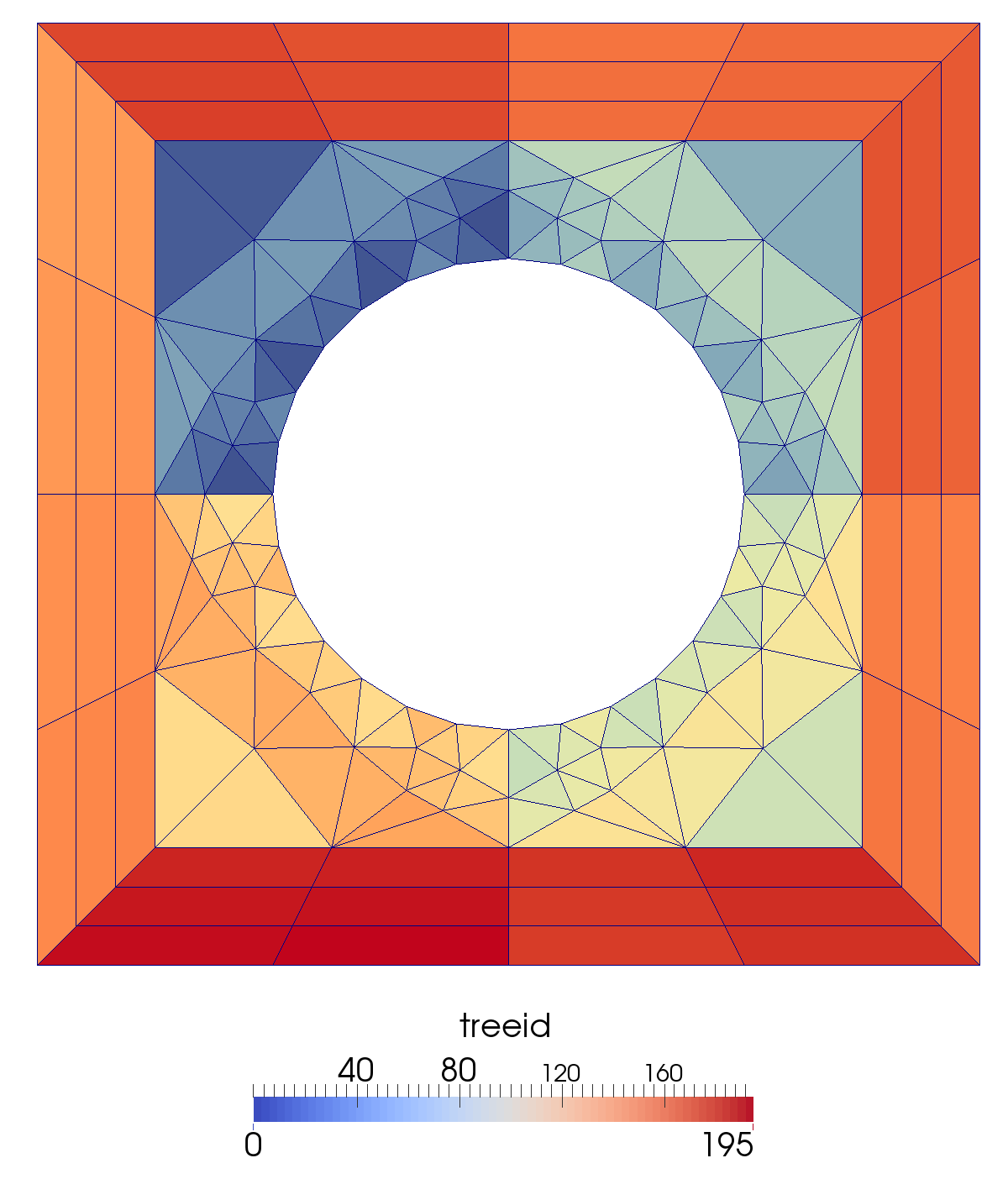}
\includegraphics[width=0.45\textwidth]{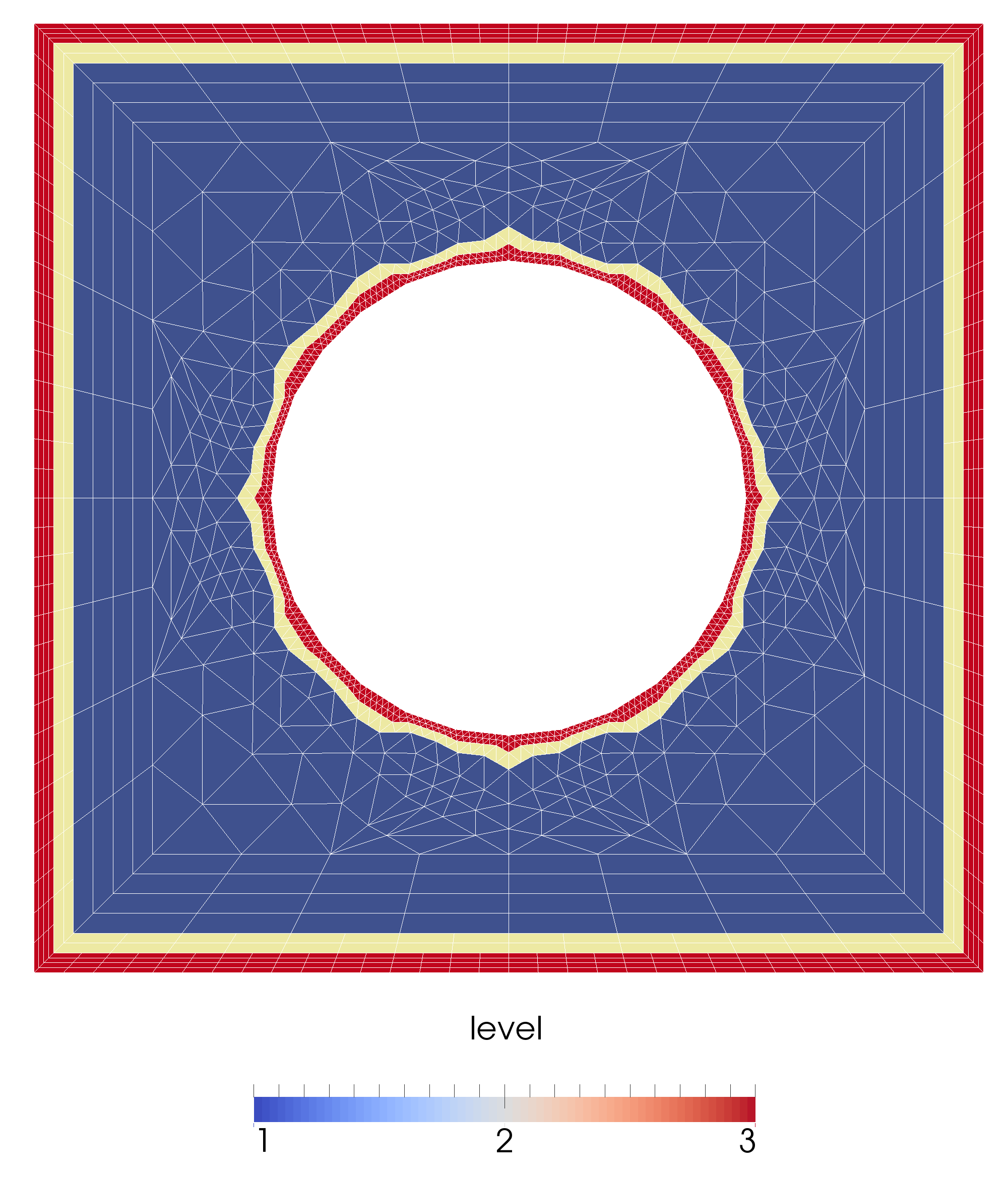}
\caption{A mesh consists of two structures, the coarse mesh (left) that represents
the connectivity of the domain, and the forest mesh (right) that consists of the
leaf elements of a refinement and is used for computation.
In this example the domain is a unit square with a circular hole.
The color coding in the coarse mesh displays each tree's unique and consecutive
identifier, while the color coding in the forest mesh represents the refinement
level of each element.
In this example we choose an initial global level 1 refinement and a
refinement of up to level 3 along the domain boundary.}
\label{fig:hybridmesh}
\end{figure}

\subsection{The tree types and their connectivity}

The trees of the coarse mesh can be of arbitrary type as long as they 
are all of the same dimension and fit together along their faces.
In particular, we identify the following tree types:
\begin{itemize}
  \item Points in 0D.
  \item Lines in 1D.
  \item Squares and triangles in 2D.
  \item Cubes and tetrahedra in 3D.
  \item Prisms and pyramids in 3D. 
\end{itemize}
Coarse meshes consisting solely of prisms or pyramids are quite uncommon; these
tree types are used primarily to transition between cubes and tetrahedra in
hybrid meshes.
The choice of SFC affects the ordering of the leaves in the forest mesh and
thus the parallel partition of elements.
Possibilities include, but are not limited to, the Hilbert, Peano, or Morton
curves for cubes and squares \cite{Peano90, Sierpinski12, Hilbert91,
WeinzierlMehl11} as well as the Sierpi\'nski or tetrahedral Morton curves for
tetrahedra and triangles \cite{BaderZenger06,BursteddeHolke16}.
In the \tetcode software used for the demonstrations in this paper we have so
far implemented Morton SFCs for cubes and squares via the \pforest library
\cite{Burstedde15} and the tetrahedral Morton SFC for tetrahedra and triangles;
other schemes may be added in a modular fashion.

\subsection{Encoding of face-neighbors}
\label{sec:faceneigh}
The connectivity information of a coarse mesh includes the neighbor
relation between adjacent trees.
Two trees are considered neighbors if they share at least one lower
dimensional face (vertex, face, or edge).
Since all of this connectivity information can be concluded from codimension-1
neighbors, we restrict ourselves to those, denoting them uniformly by
face-neighbors.
At this point we do not consider neighbor relations across lower dimensional
tree faces, as we believe without loss of generality for the partitioning
algorithms to follow.

An application often requires a quick mechanism to access the face-neighbors of
a given forest mesh element.
If this neighbor element is a member of the same tree the computation can be
carried out via the SFC logic, which involves only few bitwise operations for
the cubical and tetrahedral Morton curves \cite{Morton66, SundarSampathBiros08,
BursteddeWilcoxGhattas11, BursteddeHolke16}.
If, however, the neighbor element belongs to a different tree, 
we need to identify this tree, given the parent tree of the original element
and the tree face at which we look for the neighbor element.
It is thus advantageous to store the face-neighbors of each tree in an array
that is ordered by the tree's faces.
To this end, we fix the enumeration of faces and vertices relative to each
other as depicted in Figure \ref{fig:vertexface}.
\begin{figure}
   \center
    \includegraphics{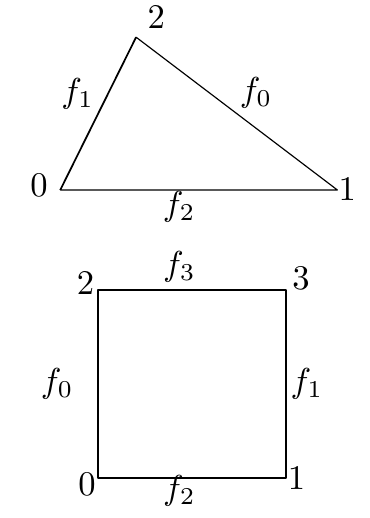}
    \includegraphics{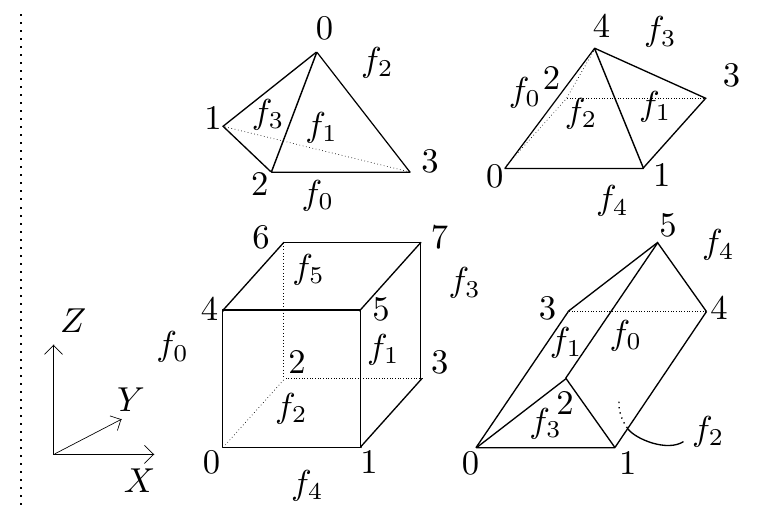}
   \caption{The vertex and face labels of the 2D (left) and 3D (right) tree types.}%
   \label{fig:vertexface}%
\end{figure}%

\subsection{Orientation between neighbors}
\label{sec:orientation}

In addition to the global index of the neighbor tree across a face, we
describe how the faces of the tree and its neighbor are rotated relative to
each other.
We allow all connectivities that can be embedded in a compact 3-manifold in
such a way that each tree has positive volume.
This includes the Moebius strip and Klein's bottle and also quite exotic
meshes, e.g.\ a cube whose one face connects to another in some rotation.
We obtain two possible orientations of a line-to-line connection, three for a
tri\-angle-to-triangle- and four for a square-to-square connection.

We would like to encode the orientation of a face connection analogously to the
way it is handled in \pforest:
At first, given a face $f$, its vertices are a subset of the vertices of the
whole tree.
If we order them accordingly and re-numerate them consecutively starting from
zero, we obtain a new number for each vertex that depends on the face $f$.
We call it the \emph{face corner number}.
If now two faces $f$ and $f'$ meet, the corner $0$ of the face with the
smaller face number is identified with a face corner $k$ in the other face. In
\pforest this $k$ is defined to be the orientation of the face connection.

In order for this operation to be well-defined, it must not depend on the
choice of the first face when the two face numbers are the same, which is
easily verified for a single tree type.
When two trees of different types meet, we generalize to
determine which face is the first one.

\begin{definition}
 We impose a semiorder on the 3-dimensional tree types as follows:
 \begin{equation}
 \begin{gathered}
 \xymatrix@C=4ex@R=1ex{ {Hexahedron} \ar@{}[rd]|-(0.6){\mathrel{\rotatebox{-20}{$<$}}} & \\
      & {Prism} \ar@{}[r]|-{<} & {Pyramid.} \\
      {Tetrahedron} \ar@{}[ru]|-{\mathrel{\rotatebox{20}{$<$}}} &
      }
 \end{gathered}
\end{equation}
\end{definition}
This comparison is sufficient since a hexahedron and a tetrahedron can never
share a common face.
We use it as follows.
\begin{definition}
 Let $t$ and $t'$ denote the tree types of two trees that meet at a common face
 with respective face numbers $f$ and $f'$.
 Furthermore, let $\xi$ be the face corner number of $f'$ matching corner $0$ of $f$ and
 $\xi'$ the face corner number of $f$ matching corner $0$ of $f'$.
 We define the orientation of this face connection as
 \begin{equation}
  \texttt{or} \adef \left\lbrace
  \begin{array}{ll}
    \xi  &  \text{if $t < t'$ or ($t = t'$ and $f \leq f'$)}, \\
    \xi' &  \text{otherwise}.
  \end{array}
  \right.
 \end{equation}
\end{definition}
We now encode the face connection in the expression $\texttt{or} \cdot F + f'$
from the perspective of the first tree and $\texttt{or} \cdot F + f$ from the
second, where $F$ is the maximal number of faces over all tree types of this
dimension.

\section{Partitioning the coarse mesh}

As outlined above, tree based AMR methods partition mesh elements with the help
of an SFC.
By cutting the SFC into as many equally sized parts as processes and assigning
part $i$ to process $i$, the repartitioning process is distributed and runs in
linear time.
Weighted partitions with a user defined weight per leaf element are also
possible and practical \cite{PinarAykanat04, BursteddeWilcoxGhattas11}.

If the physical domain has a complex shape such that many trees
are required to optimally represent it, it becomes necessary to also
partition the coarse mesh in order to reduce the memory footprint.
This is even more important if the coarse mesh does not fit into the memory of
one process, since such problems are not even computable without coarse mesh
partitioning.

\begin{figure}
\center
\includegraphics[width=\textwidth]{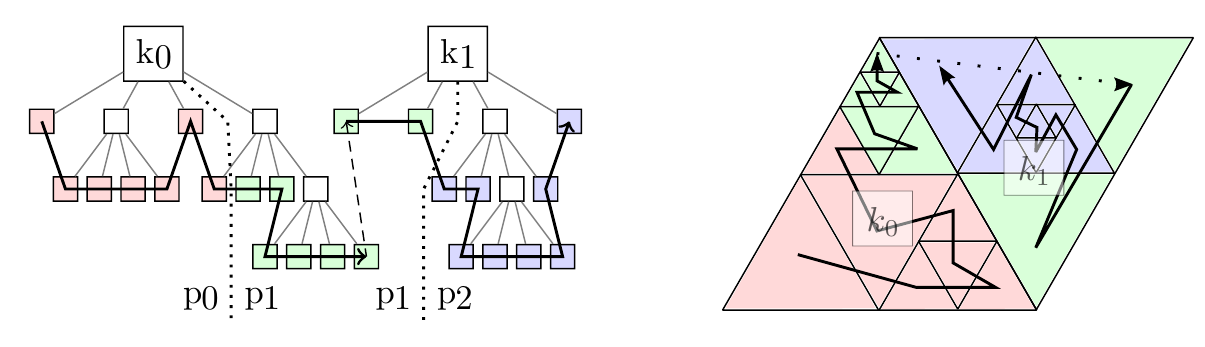}
\caption{We patch multiple trees together to model complex geometries. Here, we 
show two trees $k_0$ and $k_1$ with an adaptive refinement.
To store the forest mesh, we use a space-filling curve within each tree
and the a-priori tree order between trees.
On the left-hand side of the figure the refinement tree and its linear storage
are shown. When we partition the forest mesh to $P$ processes (here, $P=3$),
we cut the SFC in $P$ equal-size parts and assign part $i$ to process $i$.}
\label{fig:twotreeforest}
\end{figure}

Suppose the forest mesh is partitioned among the processes.
Since the forest mesh frequently references connectivity information from the
coarse mesh, a process that owns a leaf $e$ of a tree $k$ also needs the
connectivity information of the tree $k$ to its neighbor trees.
Thus, we maintain information on these so-called ghost trees.

There are two traditional approaches to partition the coarse mesh.
In the first approach
\cite{BurriDednerKloefkornEtAl06, YilmazOezturanTosunEtAl10}, the coarse mesh
is partitioned, and the owner process of a tree will own all elements of that
tree.
In other words, it is not possible for two different processes to own elements
of the same tree, which can lead to highly imbalanced forest meshes.
Furthermore, if there are less trees then processes, there will be idle
processes without any leaf elements.
In particular, this prohibits the frequently occurring special case of a single
tree.

The second approach \cite{Zumbusch03} is to first partition the forest mesh and
to deduce the coarse mesh partition from that of the forest.
If several
processes have leaf elements from the same tree, then the tree is assigned to
one of these processes, and whenever one of the other processes requests
information about this tree, communication is invoked.
This technique has the advantage that the forest mesh is load balanced much
better, but it introduces additional synchronization points in the program and
can lead to critical bottlenecks if a lot of processes request information on
the same tree.

We propose another approach, which is a variation of the second, that overcomes
the communication issue.
If several processes have leaf elements from the same tree, we duplicate this
tree's connectivity data and store a local copy of it on each of this subset
of processes.
Thus, there is no further need for communication and each process has exactly
the information it requires.
Since the purpose of the coarse mesh is not to store data that changes during
the simulation, but to store connectivity data about the physical domain the
data on each tree is persistent and does not change during the simulation.
Certainly, this concept poses an additional challenge in the (re)partitioning
process, because we need to manage multiple copies of trees without producing
redundant messages.

As an example consider the situation in Figure \ref{fig:2trees3proc}.
Here the 2D coarse mesh consists of two triangles $0$ and $1$, and the forest
mesh is a uniform level 1 mesh consisting of 8 elements. Elements $0,1,2,3$ are
the children of tree $0$ and elements $4,5,6,7$ the children of tree $1$.  If
we load-balance the forest mesh to three processes with ranks $0$, $1$ and $2$,
then a possible forest mesh partition arising from a SFC could be\\
\begin{minipage}{0.35\textwidth}
\begin{equation}
  \begin{array}{cl}
    \mathrm{rank}&\mathrm{elements}\\\hline
    0 & 0,\, 1,\, 2\\
    1 & 3,\, 4,\, 5\\
    2 & 6,\, 7,
  \end{array}
\end{equation}
\end{minipage}
\begin{minipage}{0.3\textwidth}
leading to the coarse mesh partition
\end{minipage}
\hfill
\begin{minipage}{0.3\textwidth}
\begin{equation}
  \begin{array}{cl}
    \mathrm{rank}& \mathrm{trees} \\\hline
    0&0\\
    1&0,\,1\\
    2&1.
  \end{array}
\end{equation}
\end{minipage}\\[3ex]
Thus, each tree is stored on two processes.

\begin{figure}
   \center
\begin{minipage}{0.8\textwidth}
   \center
  \includegraphics{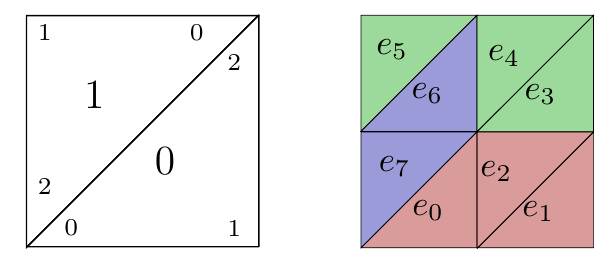}
\end{minipage}
   \caption{A coarse mesh of two trees and a uniform level 1 forest mesh.
            If the forest mesh is partitioned to three processes, each tree of the 
            coarse mesh is requested by two processes.
            The numbers in the tree corners denote the position and orientation of the tree vertices
            and the global tree ids are the numbers in the center of each tree.
            As an example SFC we take the triangular Morton curve from
            \cite{BursteddeHolke16}, but the situation of multiple trees per
            process can occur for any SFC.}
   \label{fig:2trees3proc}
\end{figure}

\subsection{Valid partitions}

We allow any SFC induced partition for the forest.
This gives us some information on the type of coarse mesh partitions that we
can expect as follows.

\begin{definition}
  In general a \emph{partition} of a coarse mesh of $K$ trees $\indexset{K}$
to $P$ processes $\indexset{P}$ is a map $f$ that assigns each process a
certain subset of the trees,
\begin{equation}
  f\colon \indexset{P} \longrightarrow \mathcal{P}\indexset{K},
\end{equation}
and whose image covers the whole mesh:
\begin{equation}
\bigcup_{p=0}^{P-1} f(p) = \indexset{K}.
\end{equation}
Here, $\mathcal{P}$ denotes the set of all subsets (power set).
We call $f(p)$ the \emph{local trees} of process $p$ and explicitly allow that
$f(p)\cap f(q)$ may be nonempty. 
If so, the trees in this intersection are \emph{shared} between processes $p$ and $q$.
\end{definition}

For a partition $f$ we denote with  ${k_p}$ the local tree on $p$ with lowest
global index and with ${K_p}$ the one with the highest global index.

We need not consider all possible partitions of a coarse mesh. Since we assume
that a forest mesh partition comes from a space-filling curve, we
can restrict ourselves to a subset of partitions.
\begin{definition}
  \label{def:valid}
  Consider a partition $f$ of a coarse mesh with $K$ trees.
  We say that $f$ is a \emph{valid partition} if there exists a forest
  mesh with $N$ leaves and a (possibly weighted) SFC partition of it that induces $f$.
  Thus, for each process $p$ and
  each tree $k$ we have $k\in f(p)$ if and only if there exists a leaf $e$ of
  the tree $k$ in the forest mesh that is partitioned to process $p$.
  Processes without any trees are possible; in this case $f(p) =
\emptyset$.  
\end{definition}

We deduce three properties that characterize valid partitions and lead to a 
definition that is independent of a forest mesh.

\begin{proposition}
  \label{prop:valid}
  A partition $f$ of a coarse mesh is valid if and only if it fulfills the following 
  properties.
\begin{enumerate}[(i)]
\item The tree indices of a process' local trees are consecutive, thus
  \begin{equation}
  f(p) = \set{{k_p},{k_p+1},\dots,{K_p}},\,\text{ or }\,f(p) =\emptyset.\label{eq:valid1}
 \end{equation}
\item A tree index of a process $p$ may be smaller than a
      tree index on a process $q$ only if $p\leq q$:
      \begin{equation}
        p \leq q  \Rightarrow  K_p \leq k_q \quad(\text{if}\,\,
        f(p)\neq\emptyset\neq f(q)).\label{eq:valid2} 
      \end{equation}
 \item The only trees that can be shared by process $p$ with other processes
   are $k_p$ and $K_p$:
   \begin{equation}
  f(p) \cap f(q) \subseteq \set{{k_p},{K_p}},\, \text{for}\,\, p\neq q.\label{eq:valid3}
   \end{equation}
\end{enumerate}
\begin{proof}
  We show the only-if direction first.
  So let an arbitrary forest mesh with SFC partition be given, such that $f$ is
  induced by it. In the SFC order the leaves are sorted according to their SFC
  indices.
  If ${(i,I)}$ denotes the leaf corresponding to the $I$-th leaf in the $i$-th
  tree and tree $i$ has $N_i$ leaves,
  then the complete forest mesh consists of the leaves
  \begin{equation}
    \set{{(0, 0)},{(0, 1)},{(0, N_0-1)},{(1, 0)},\dots,{(K-1, N_{K-1}-1)}}.
  \end{equation}
  The partition of the forest mesh is such that each process 
  $p$ gets a consecutive range
  \begin{equation}
    \set{{(k_p, i_p)},\dots,{(K_p, i'_{p})}}
  \end{equation}
  of leaves.
  Where $(k_{p+1}, i_{p+1})$ is the successor of $(K_p, i'_p)$.
  and the $k_p$ and $K_p$ form increasing sequences
  with additionally $K_p\leq k_{p+1}$.
  The coarse mesh partition is then given by
  \begin{equation}
    f(p) = \set{k_p,k_p+1,\dots,K_p},\,\, \text{for all } p,
  \end{equation}
  which shows properties $(i)$ and $(ii)$.
  To show $(iii)$ we assume that $f(p)$ has at least three elements, thus
  $f(p) = \set{{k_p},{k_p+1},\dots,{K_p}}$.
  However, this means that in the forest mesh partition each leaf that is a child
from the trees in $\set{k_p+1,\dots,K_p-1}$ is partitioned to $p$.
  Since the forest mesh partitions are disjoint, no other process can hold leaf
  elements from these trees and thus they cannot be shared.

  To show the if-direction, suppose the partition $f$ fulfills (i), (ii) and (iii).
  We construct a forest mesh with a weighted SFC partition as follows.
  Each tree that is local to a single process is not refined and thus a single
  leaf element in the forest mesh. If a tree is shared  by $m$ processes then we refine it uniformly
  until we have more then $m$ elements. It is now straightforward to choose the weights of the elements
  such that the corresponding SFC partition induces $f$.
\end{proof}
\end{proposition}

We directly conclude:

\begin{corollary}
\label{cor:valid1}
  In a valid partition each pair of processes can share at most 
  one tree, thus
  \begin{equation}
      |f(p) \cap f(q)| \leq 1
  \end{equation}
  for each $p\neq q$.
\begin{proof} 
  Suppose not then with \eqref{eq:valid3} we know that there exist
  two processes $p$ and $q$ with $p<q$ such that $f(p) \cap f(q) \agl
  \set{{k_p},{K_p}}=\set{{k_q},{K_q}}$ and $k_p \neq K_p$. Thus, $K_p > k_p = k_q$
  which contradicts property \eqref{eq:valid2}.
\end{proof}
\end{corollary}
\begin{corollary}
\label{cor:valid2}
  If in a valid partition $f$ of a coarse mesh the tree $k$ is shared
  between processes $p$ and $q$, then for each $p < r < q$
  \begin{equation}
    f(r) = \set{k} \quad\mathrm{or}\quad f(r) = \emptyset.
  \end{equation}
\end{corollary}
\begin{proof}
We can directly deduce this from \eqref{eq:valid1}, \eqref{eq:valid2} and
Corollary \ref{cor:valid1}.
\end{proof}
In order to properly deal with empty processes in our calculations, we
define start and end tree indices for these as well.
\begin{definition}
  \label{def:seforempty}
  Let $p$ be an empty process in a valid partition $f$, thus $f(p)=\emptyset$.
  Let furthermore $q<p$ be maximal such that $f(q)\neq\emptyset$.
  Then we define the start and end indices of $p$ as
  \begin{align}
    k_p&:=K_q+1,\\
    K_p&:=K_q = k_p-1.
  \end{align}
  If no such $q$ exists, then no rank lower than $p$ has local trees and 
  we set $k_p=0$, $K_p=-1$.
  With these definitions, equations \eqref{eq:valid1} and
  \eqref{eq:valid2} are valid if any of the processes are empty.
\end{definition}

From now on all partitions in this manuscript are considered as valid even
if not stated explicitly.

\subsection{Encoding a valid partition}

A typical way to define a partition in a tree based code is to store an array
\texttt{O} of tree offsets for each process, that is, for process $p$ the global
index of the first local tree is stored.  
The range of local trees for process $p$ can then be computed as
$\set{{\texttt O[p]},\dots,{\texttt O[p+1]-1}}$.
However, for valid partitions in the coarse mesh setting, this information would not
be sufficient because we would not know which trees are shared.
We thus slightly modify the offset array by adding a negative sign when the first tree 
of a process is shared.
\begin{definition}
  \label{def:offsetarray}
  Let $f$ be a valid partition of a coarse mesh with $k_p$ being the 
  index of $p$'s first local tree.
  Then we store this partition in an array \texttt O of length $P+1$, where
  for $0 \leq p < P$
  \begin{equation}
    \label{eq:offsetarray}
 \texttt O[p] = \left\lbrace \begin{array}{ll}
      k_p, &
      \begin{array}{l}
        \textrm{if ${k_p}$ is not shared with the next smaller}\\
        \textrm{nonempty process or }f(p)=\emptyset,
      \end{array}
     \\[2ex]
               -k_p - 1, & \,\,\,\textrm{if it is.}
                   \end{array}\right.
  \end{equation}
  Furthermore, in $\texttt O[P]$ we store the total number of trees.
\end{definition}
Because of the definition of $k_p$ we know that $\texttt O[0] = 0$ for all
valid partitions.

\begin{lemma}
  \label{lem:rangefromoffset}
 Let $f$ be a valid partition and \texttt O as in Definition \ref{def:offsetarray}.
 Then
  \begin{equation}
    \label{eq:getfirstfromoffset}
    k_p =  \begin{cases}
      \texttt O[p], &\textrm{if }\texttt O[p] \geq 0,\\[1ex]
   |\texttt O[p] + 1|, & \textrm{if }\texttt O[p] < 0,
     \end{cases}
  \end{equation}
  and
  \begin{equation}
    \label{eq:getlastfromoffset}
    K_p = |\texttt O[p+1]| - 1.
  \end{equation}
\end{lemma}
  \begin{proof}
    The first statement follows since equations \eqref{eq:offsetarray} and \eqref{eq:getfirstfromoffset}
    are inverses of each other.

    For equation \eqref{eq:getlastfromoffset} we distinguish two cases:
    First, let $f(p)$ be nonempty.
    If the last tree of $p$ is not shared with $p+1$, then
    it is ${k_{p+1}} - 1$ and $\texttt O[p+1] = k_{p+1}$, thus we have 
    \begin{equation}
      K_p = k_{p+1}-1 = |\texttt O[p+1]| - 1.
    \end{equation}
    If the last tree of $p$ is shared with $p+1$, then it is
    $k_{p+1}$,
    the first local tree of $p+1$ and thus $\texttt O[p+1] = -k_{p+1}-1$ and 
    \begin{equation}
      K_p= k_{p+1} = |-k_{p+1}| = |\texttt O[p+1]| - 1.  
    \end{equation}
    Let now $f(p)=\emptyset$.
    If $k_{p+1}$ is not shared, then $k_{p+1}=k_p=K_p+1$ by Definition \ref{def:seforempty}
    and $\texttt O[p+1]=k_{p+1}$ by \eqref{eq:offsetarray}.
    Thus,
    \begin{equation}
      K_p = k_p - 1 = k_{p+1}-1 = |\texttt O[p+1]| - 1.
    \end{equation}
    If $k_{p+1}$ is shared, then again by Definition \ref{def:seforempty}
    $k_{p+1}=k_p-1=K_p$ and $\texttt O[p+1] = -k_{p+1} - 1$, such that we 
    obtain
    \begin{equation}
      K_p = k_{p+1} = k_{p+1} + 1 - 1 = |\texttt O[p+1]| - 1.
    \end{equation}
  \end{proof}
\begin{corollary}
\label{col:numloctrees}
  In the setting of Lemma \ref{lem:rangefromoffset} the number $n_p$ of local
trees of process $p$ fulfills
  \begin{equation}
    n_p = |\texttt O[p+1]| - k_p = \begin{cases}
                             |\texttt O[p+1]|-\texttt O[p],& \textrm{if }\texttt O[p]\geq 0\\
                             |\texttt O[p+1]|-|\texttt O[p]+1|,& \textrm{else.}
                           \end{cases}
  \end{equation}
\end{corollary}
  \begin{proof}
    This follows from the formula $n_p = K_p - k_p + 1$.
  \end{proof}
Lemma \ref{lem:rangefromoffset} and Corollary \ref{col:numloctrees} show that
for valid partitions the array \texttt O carries the same information as the
partition $f$.

\subsection{The ghost trees}

A valid partition gives the information about the local trees of a process.
These trees are all trees from whom a forest has local elements.
In many applications it is necessary to create a layer of ghost (or halo) elements
of the forest to properly exchange data with the neighboring processes.
Since these ghost elements may be descendants of non local trees, we also want
to store those trees as ghost trees.
To be independent of a forest mesh we store each possible ghost tree and do not
consider whether there are actually forest mesh ghost elements in this tree.
This, however, does only affect the first and the last local tree, where we possibly store
more ghosts than needed by a forest.
Since we restrict the neighbor information to face-neighbors, we also restrict
ourselves to face-neighbor ghosts in this paper.
\begin{definition}
  Let $f$ be a valid partition of a coarse mesh.
  A \emph{ghost tree} of a process $p$ is any tree $k$ such that
  \begin{itemize}
    \item $k\notin f(p)$, and
    \item There exists a face-neighbor $k'$ of $k$, such that $k'\in f(p)$.
  \end{itemize}
\end{definition}

If a coarse mesh is partitioned according to $f$ then each process $p$ will
store its local trees and its ghost trees.

\subsection{Computing the communication pattern}
Suppose we are in a setting where a coarse mesh is partitioned among the
processes $p\in \indexset P$ according to a partition $f$.
The input of the partition algorithm is this coarse mesh and a second partition
$f'$, and the output is a coarse mesh that is partitioned according to the
second partition.

We suppose that besides its local trees and ghosts each process knows the
complete partition tables $f$ and $f'$, for example in the form of offset arrays.
The task is now for each process to identify the processes that it needs to send 
local trees and ghosts to and carry out this communication.
A process also needs to identify the processes it receives trees and ghosts
from and do the receiving.
We discuss here how each process can compute this information from the
offset arrays without further communication.

\subsubsection{Ownership during partition}
The fact that trees can be shared between multiple processes poses a challenge 
when repartitioning a coarse mesh.
Suppose we have a process $p$ and a tree $k$ with $k\in f'(p)$,
and $k$ is a local tree for more than one process in the partition $f$.
We do not want to send the tree multiple times, so how do we decide which
process sends $k$ to $p$?

A simple solution would be that the process with the smallest index that $k$ is
a local tree to sends $k$.
This process is unique and it can be determined without communication.
However, suppose that the two processes $p$ and $p-1$ share the tree $k$
in the old partition and $p$ will also have this tree in the new partition.
Then $p-1$ would send the tree to $p$ even though we could handle this
situation without any communication.

We solve this by only sending a local tree to a process $p$ if
this tree was not already local on $p$:
\begin{pardgm}
  \label{par:sendtrees}
 In a repartitioning setting with $k\in f'(p)$, the process that sends $k$ to
$p$ is \begin{itemize}
   \item $p$, if $k$ already is a local tree of $p$, or else
   \item $q$, with $q$ minimal such that $k\in f(q)$.
 \end{itemize}
\end{pardgm}
We acknowledge that sending from $p$ to $p$ in the first case is just a local
data movement involving no communication.

\begin{definition}
In a repartitioning setting, given a process
$p$ we define the sets $S_p$ and $R_p$ of processes
that $p$ sends local trees to, resp.\ receives from, thus
\begin{subequations}
\begin{align}
 S_p&:=\set{0\leq p < P \abst{|} p \textrm{ sends local trees to p'}},\\[1ex]
 R_p&:=\set{0\leq p < P \abst{|} p \textrm{ receives local trees from p'}}.
\end{align}
\end{subequations}
These sets can both include the process $p$ itself.
Furthermore, we establish notations for the smallest and biggest ranks in 
these sets:
\begin{subequations}
\begin{align}
 s_\mathrm{first}:= \min S_p,&\quad
s_\mathrm{last}:= \max S_p,\\
 r_\mathrm{first}:= \min R_p,&\quad
r_\mathrm{last}:= \max R_p.
\end{align}
\end{subequations}
\end{definition}

$S_p$ and $R_p$ are uniquely determined by Paradigm \ref{par:sendtrees}. 

\subsubsection{An example}
We discuss a small example, see Figure \ref{fig:partitionex}.  Here, we
repartition a partitioned coarse mesh of five trees among three processes.
We color the local trees of process $0$ in blue, the local trees of process 
$1$ in green and the local trees of process $2$ in red.
The initial partition $f$ is given by
\begin{equation}
\label{eq:example_start}
\texttt{O}=\set{0, -2, 3, 5},
\end{equation}
and the new partition $f'$ by
\begin{equation}
\texttt{O'} = \set{0, -3, -4, 5}.
\end{equation}
Thus, initially tree $1$ is shared by processes $0$ and $1$ and in the new
partition tree $2$ is shared by processes $0$ and $1$ and tree $3$ by processes
$2$ and $3$.
We give the local trees that each process will send to each other process
in a table, where the set in row $i$ column $j$ is the set of local trees that
process $i$ send to process $j$.
\begin{equation}
\begin{array}{c|ccc}
  & 0 & 1 & 2 \\ \hline
0 & \set{0,1} & \emptyset & \emptyset  \\
1 & \set{2} & \set{2} & \emptyset \\
2 & \emptyset & \set{3} & \set{3,4}
\end{array}
\end{equation}
This leads to the sets $S_p$ and $R_p$:

\begin{subequations}
\label{eq:example_end}
\begin{align}
S_0 &= \set{0},   & R_0&=\set{0,1},\\
S_1 &= \set{0,1}, & R_1&=\set{1,2},\\
S_2 &= \set{1,2}, & R_2&=\set{2}.
\end{align}
\end{subequations}

For example process $1$ keeps the tree $2$ that is also needed by process
$0$. Thus, process $1$ sends tree $2$ to process $0$. Process $0$ also needs
tree $1$, which is local on process $1$ in the old partition. But, since it is
also local to process $0$, process $1$ does not send it.

{
\begin{figure}
\center
\includegraphics[width=0.6\textwidth]{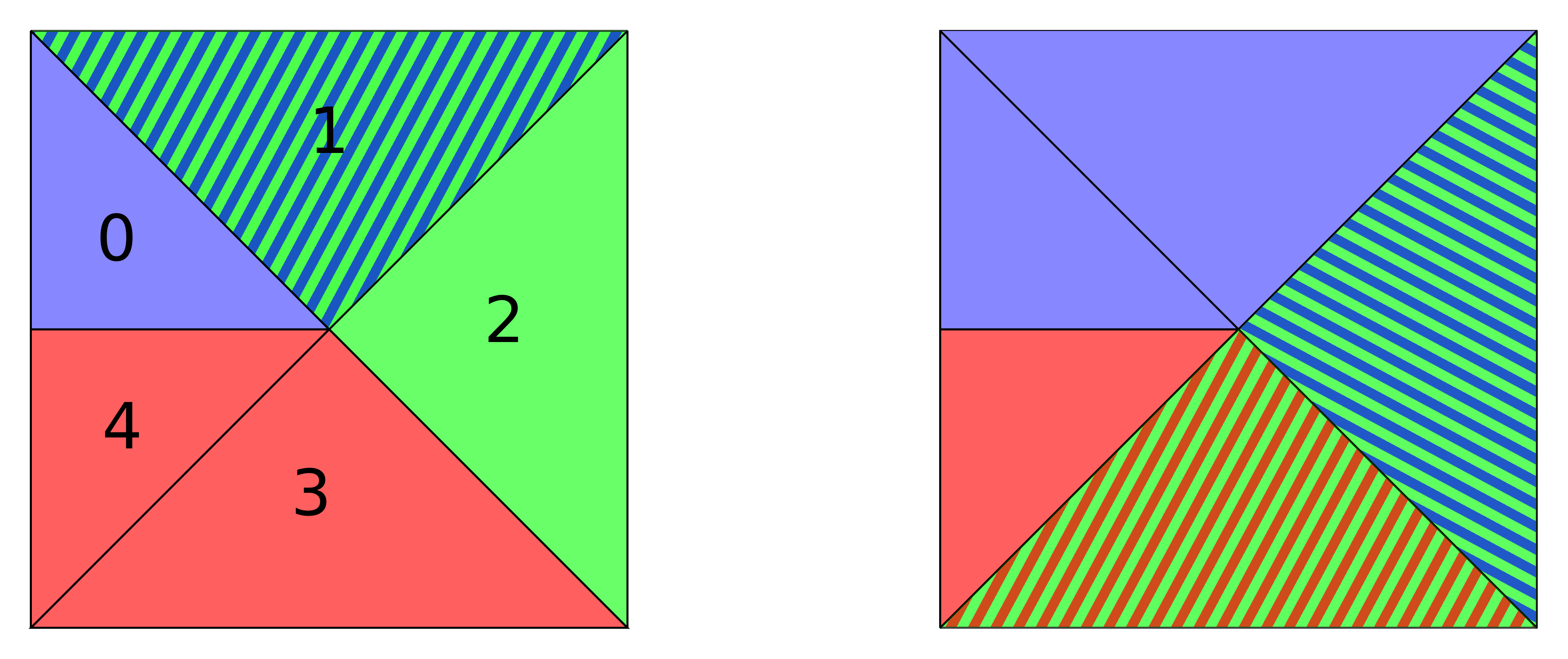}
\caption{A small example for coarse mesh partitioning. The coarse mesh consists
of 5 trees and is partitioned among three processes. Left: The initial partition
\texttt{O}. Right: The new partition \texttt{O'}.
The colors of the trees encode the processes that have the respective tree as
local tree. Process $0$ in blue, process $1$ in green, and process $2$ in red.
Thus, initially tree 1 is shared among processes $0$ and $1$ and in the new
partition tree 2 is shared among processes $0$ and $1$, and tree 3 among 
processes $1$ and $2$.
In equations \eqref{eq:example_start}-\eqref{eq:example_end} we depict the
sets \texttt{O} and \texttt{O'}, the trees that each process sends,
and the sets $S_p$ and $R_p$.}
\label{fig:partitionex}
\end{figure}
}

\subsubsection{Determining $S_p$ and $R_p$}
\label{sec:detSpRp}

In this section we show that each process can compute the sets $S_p$ and 
$R_p$ from the offset array without further communication.
It will become clear in Section~\ref{sec:partghosknow} that
ghost trees need not yet be discussed at this point.

\begin{proposition}
\label{prop:commpattern}
A process $p$ can calculate the sets $S_p$ and $R_p$ without further
communication from the offset arrays of the new and old partition.
Once the first and last element of each set are known, $p$ can determine
in constant time whether any given rank is in any of those sets.
\end{proposition}

We split the proof in two parts. First we discuss how a process can compute
$s_\mathrm{first},\, s_\mathrm{last},\, r_\mathrm{first},$ and $r_\mathrm{last}$,
and then how it can decide for two processes $\tilde p$ and $q$ whether $q\in S_{\tilde p}$.

We derive $S_p$ from $s_\mathrm{first}$ and $s_\mathrm{last}$.
To determine $s_\mathrm{first}$ we consider two cases.
First, if the first local tree of $p$ is not shared with a smaller rank,
then $s_\mathrm{first}$ is the smallest process that has this tree in the new
partition and did not have it in the old one.
We can find this process with a binary search in the offset array.

Second, if the first tree of $p$ is shared with a smaller rank, then $p$ only
sends it in the case that $p$ keeps this tree in the new partition.
Then $s_\mathrm{first} = p$.
Otherwise, we consider the second tree of $p$ and proceed with a binary search
as in the first case.

To compute $s_\mathrm{last}$ we notice that among all ranks having
$p$'s old last tree in the new partition that did not already have it before,
$s_\mathrm{last}$ is the biggest
(except when $p$ itself is this biggest rank, in which case it certainly had
the last tree before).
We can determine this rank with a binary search as well.
If no such process exists, we proceed with the second last tree of $p$.

\begin{remark}
The special case when $S_p=\emptyset$ occurs under
three circumstances:
\begin{enumerate}
\item $p$ does not have any local trees.
\item $p$ only has one local tree, it is shared with a smaller rank,
      and $p$ does not have this tree as a local tree in the new partition.
\item $p$ only has two trees, for the first one case 2 holds and
      the second ($=$last) one is shared with a set $Q$ of bigger ranks.
      There is no process $q\notin Q$ that has this tree as a local tree
      in the new partition.
\end{enumerate}
These conditions can be queried before computing $s_\mathrm{first}$ and
$s_\mathrm{last}$.
\end{remark}

Similarly, to compute $R_p$, we first look at the smallest and biggest elements
of this set.
These are the first and last process that $p$ receives trees
from.
$r_\mathrm{first}$ is the smallest rank that had $p$'s new first tree as a
local tree in the old partition, or it is $p$ itself if this tree was also a
local tree on $p$.
And $r_\mathrm{last}$ is the smallest rank bigger then $r_\mathrm{first}$
that had $p$'s new last local tree as a local tree in the old partition, or $p$ itself.
We can find both of these with a binary search in the offset array of
the old partition.

\begin{remark}
$R_p$ is empty if and only if $p$ does not have any local trees in the new
partition.
\end{remark}

\begin{lemma}
\label{lem:Spconstant}
Given any two processes $\tilde p$ and $q$ the process $p$ can determine in
constant time whether $q\in S_{\tilde p}$.
In particular, this includes the cases $\tilde p = p$ and $q = p$.
\begin{proof}
 Let $\hat k_{\tilde p}$ be the first non-shared local tree of $\tilde p$ in the
 old partition. Let $\hat K_{\tilde p}$ be the last local tree of $\tilde p$ in the old
 partition if it is not the first local tree of $q$ in the old partition.
 Let it be the second last local tree otherwise.
 Furthermore, let $\hat k_q$ and $\hat K_q$ be the first, resp.\ last, local trees of
 $q$ in the new partition. We add $1$ to $\hat k_q$ if $q$ sends its first local
 tree to itself and this tree is also the new first local tree of $q$.
 We claim that $q\in S_{\tilde p}$ if and only if all of the four inequalities
 \begin{align}
  \hat k_{\tilde p} \leq \hat K_{\tilde p},\quad
  \hat k_{\tilde p} \leq \hat K_q,\quad
  \hat k_q \leq \hat K_{\tilde p},\quad\textrm{and}\quad
  \hat k_q\leq \hat K_q\phantom{,}
 \end{align}
 hold.
 The only-if direction follows, since if $\hat k_{\tilde p}>\hat K_{\tilde p}$, then
$\tilde p$ does not have trees to send to $q$. If $\hat k_{\tilde p} > \hat K_q$, 
 then the last new tree on $q$ is smaller than the first old tree on $\tilde p$.
 If $\hat k_q > \hat K_{\tilde p}$ then the last tree that $\tilde p$ could send
 is smaller than the first new local tree of $p$. And if $\hat k_q > \hat K_q$ then 
 $q$ does not receive any trees from other processes.
 Thus $\tilde p$ cannot send trees to $q$ if any of the four conditions is not
 fulfilled.
 The if-direction follows, since if all four conditions are fulfilled there
 exist at least one tree $k$ with 
 \begin{align}
  \hat k_{\tilde p}\leq k \leq \hat K_{\tilde p} \quad\mathrm{and}\quad
  \hat k_q\leq k \leq K_q.
 \end{align}
 Any tree with this properties is sent from $\tilde p$ to $q$.
 Process $p$ can compute the four values $\hat k_{\tilde p}, \hat K_{\tilde p}, \hat k_q$ and 
 $\hat K_q$ 
 from the partition offsets in constant time.
\end{proof}
\end{lemma}

\begin{remark}
Let $p$ be a process that is not empty in the new partition.
For symmetry reasons,
$R_p$ contains exactly those processes $\tilde p$ with $r_\mathrm{first}\leq
\tilde p \leq r_\mathrm{last}$ and $p\in S_{\tilde p}$.
\end{remark}

Thus, in order to compute $S_p$, we can compute $s_\mathrm{first}$ and
$s_\mathrm{last}$ and then check for each rank $q$ in between whether the conditions of
Lemma \ref{lem:Spconstant} are fulfilled with $\tilde p = p$.
For each process this check only takes constant run time.

Now, to compute $R_p$ we can compute $r_\mathrm{first}$ and $r_\mathrm{last}$, and then
check for each rank $q$ in between whether $p\in S_q$.

These considerations provide the proof for Proposition \ref{prop:commpattern}.

\subsection{Face information for ghost trees}
\label{sec:partghosknow}

We identify five different types of possible face connections in a coarse
mesh:\\[1ex]
\begin{minipage}{0.4\textwidth}
\begin{enumerate}
  \item Local tree to local tree.
  \item Local tree to ghost tree.
  \item Ghost tree to local tree.
\end{enumerate}
\end{minipage}
\begin{minipage}{0.5\textwidth}
\begin{enumerate}
  \setcounter{enumi}{3}
  \item Ghost tree to ghost tree.
  \item Ghost tree to non-local and non-ghost tree.
\end{enumerate}
\end{minipage}\\[1ex]
There are several possible approaches to which of these face connections of a local coarse
mesh we could actually store. As long as each face connection between any two
neighbor trees is stored at least once globally, the information of the coarse
mesh over all processes is complete, and a single process could reproduce all
five types of face connection at any time, possibly using communication.
Depending on which of these types we store, the pattern for sending and
receiving ghost trees during repartitioning changes.
Specifically, the trees that will become a ghost on the receiving process
may be either a local tree or a ghost on the sending process.

When we use the maximum possible information of all five connections, we
have the most data available and can minimize the communication required.
In particular, from the non-local neighbors of a ghost and 
the partition table a process can compute which other processes this ghost
is also a ghost of and of which it is a local tree. With this information we can
ensure that a ghost is only sent once and only from a process that also sends
local trees to the receiving process.

The outline of the sending/receiving phase of the partition algorithm then
looks like this:

\begin{enumerate}
 \item For each $q\in S_p$: Send local trees that will be owned by $q$
   (following Paradigm \ref{par:sendtrees}).  
 \item Consider sending a neighbor of these trees to $q$ if
   it will be a ghost on $q$.
   Send one of these neighbors if both
 \begin{itemize}
   \item $p$ is the smallest rank among those that consider sending this neighbor as a ghost, and
   \item $p\neq q$ and $q$ does not consider sending this neighbor as a ghost to itself.
 \end{itemize}
 \item For each $q\in R_p$: Receive the new local trees and ghosts from $q$.
\end{enumerate}
In step 2 a process needs to know, given a ghost that is considered for
sending to $q$, which other processes consider sending this ghost to $q$.
This can be calculated without further communication from the face-neighbor information
of the ghost. Since we know for each ghost the global index of each of its
neighbors, we can check whether any of these neighbors is currently local on a
different process $\tilde p$ and will be sent to $q$ by $\tilde p$. If so, we
know that $\tilde p$ considers sending this ghost to $q$.

Using this method, each local tree and ghost is sent only once to each receiver,
and only those processes send that send local trees anyway,
thus we have a minimum number of communications and data movement.
Storing less information would either increase the number of communicating 
processes or the amount of data that is communicated.

Suppose we would not store the face connection type 5, thus for ghost trees 
we do not have the information with which non-local trees it is connected.
With this face information we can use a communication pattern such that
each ghost is only received once by a process $q$, by sending the new ghost
trees from a process that currently has it as a local tree (taking Paradigm
\ref{par:sendtrees} into account). However, this designated process might not
be an element of $R_q$, in which case additional processes would communicate.

If we only store the local tree face information, types 1 and 2, then
we have minimal control over the ghost face connections.
Nevertheless, we can define the partition algorithm by specifying that if a 
process $p$ sends local trees to a process $q$, it will send all neighbors of
these local trees as potential ghosts to $q$.
The process $q$ is then responsible for deleting those trees that it received
more than once.
With this method the number of communicating processes would be the same but
the amount of data communicated would increase.

\begin{figure}
  \center
  \begin{minipage}{0.8\textwidth}
    \includegraphics{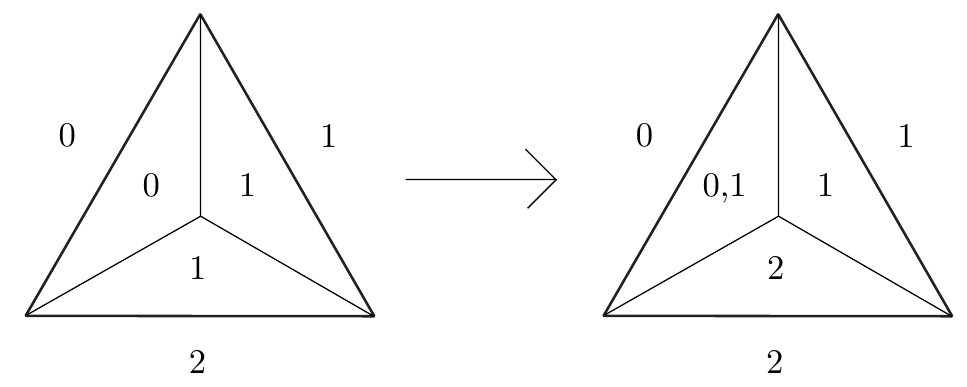}
  \end{minipage}\\[4ex]

  \begin{minipage}{\textwidth}
    \begin{tabular}{|c||c|c|c||c|c|c||c|c|c|}
      \hline
     &\multicolumn{3}{c||}{1, 2} &
      \multicolumn{3}{c||}{1, 2, 3, 4} &
      \multicolumn{3}{c||}{1, 2, 3, 4, 5} \\
      \hline
     p &0&1&2          & 0&1&2&            0&1&2\\ \hline 
     0&0(1,2)&0(2)&--- & 0&0&(0)          &0(1,2) & 0 & ---\\[0.8ex]
     1&---&1(2)&2(0,1) & (1,2) & 1(2) & 2(1)&--- & 1(2) & 2 (0,1)\\

      \hline
    \end{tabular}
  \end{minipage}
  \caption{Repartitioning example of a coarse mesh showing the communication
  patterns controlled by the amount of face information available.
  Top: A coarse mesh of three trees is repartitioned. The numbers outside of the
  trees are their global indices. The numbers inside of each tree
  denote the processes that have this tree as a local tree.
  At first process 0 has tree 0, process 1 trees 1 and 2, and process 2 no
  trees as local trees.  After repartitioning process 0 has tree 0, process 1
  trees 0 and 1, and process 2 tree 2 as local trees.
  Bottom: The table shows for each usage of face connection types which processes
  send which data.  The row of process $i$ shows in column $j$ which local
  trees $i$ sends to $j$, and---in parentheses---which ghosts it sends to $j$.
  Using face connection types 1--4 we use more communication partners (process
  0 sends to process 2 and process 1 to process 0) than with all five types.
  Using types 1 and 2 only, duplicate data is sent (process 0 and
  process 1 both send the ghost tree 2 to process 1).}
  \label{fig:comparelevels}
  \end{figure}

  We give an example comparing the three face information strategies in Figure
\ref{fig:comparelevels}.
  To minimize the communication and overcome the need for postprocessing
  steps, it is thus recommended to store all five types of face connections.

\section{Implementation}

Let us begin by outlining the main data structures for trees, ghosts, and the
coarse mesh, and continue with a section on how to update the local tree and
ghost indices.
After this we present the partition algorithm to repartition a given coarse mesh
according to a pre-calculated partition array.
We emphasize that the coarse mesh data stores pure connectivity.
In particular, it does \emph{not} include the forest information, that is leaf
elements and per-element payloads, which are managed by separate, existing
algorithms.

\subsection{The coarse mesh data structure}

Our data structure \texttt{cmesh} that describes a (partitioned) coarse mesh
has the entries:
\begin{itemize}
 \item \texttt{O} --- An array storing the current partition table, see
   Definition \ref{def:offsetarray}.
 \item $n_p$ --- The number of local trees on this process.
 \item $n_{\mathrm{ghosts}}$ --- The number of ghost trees on this process.
 \item \texttt{trees} --- A structure storing the local trees in order of their global index.
 \item \texttt{ghosts} --- A structure storing the ghost trees in no particular order.
\end{itemize}

We use $32$-bit signed integers for the local tree counts in \texttt{trees} and
\texttt{ghosts} and $64$-bit integers for global counts in \texttt O.
This limits the number of trees per process to $2^{31}-1\cong 2\times 10^9$.
However, even with an overly optimistic memory usage of only $10$ bytes per tree,
storing that many trees would require about $18.6$GB of memory per process.
Since on most distributed machines the memory per process is indeed much smaller,
restricting to $32$-bit integers does not effectively limit the local number of 
trees.
In presently unimaginable cases, we can still switch to $64$-bit integers.

  We call the index of a local tree inside the \texttt{trees} array the \emph{local index}
  of this tree.
  Analogously,
  we call the index of a ghost in \texttt{ghosts} the \emph{local index} of that ghost.
On process $p$, we compute the global index $k$ of a tree in \texttt{trees}
from its local index $\ell$ and the global index $k_p$ of the first local tree
and vice versa, since
\begin{equation} 
  k = k_p + \ell
  .
\end{equation}
This allows us to address local trees with their local indices
using $32$-bit integers.

Each \texttt{tree} in the array \texttt{trees} stores the following data:
\begin{itemize}
  \item \texttt{eclass} --- The tree's type as a small number (triangle,
    quadrilateral, etc.).
 \item \texttt{tree\_to\_tree} --- An array storing the local tree and ghost
   neighbors along this tree's faces. See Section \ref{sec:faceneigh}.  
 \item \texttt{tree\_to\_face} --- An array encoding for each face the neighbor
   face number and the orientation of the face connection.
   See Section \ref{sec:orientation}.
 \item \texttt{tree\_data} --- A pointer to additional data that we store for the
   tree, for example geometry information or boundary conditions defined by an
   application.
\end{itemize}

The $i$-th entry of \texttt{tree\_to\_tree} array encodes the tree number of the
face-neighbor at face $i$ using an integer $k$ with $0\leq k <
n_p+n_{\mathrm{ghosts}}$.  If $k<n_p$, the neighbor is the local tree with
local index $k$.  Otherwise the neighbor is the ghost with local
index $k - n_p$.

We do not allow a face to be connected to itself. Instead, we use such a
connection in the face-neighbor array to indicate a domain boundary.
However, a tree can be connected to itself via two different faces.
This allows for one-tree periodicity, as say in a 2D torus consisting of a
single square tree.

Each \texttt{ghost} in the array \texttt{ghosts} stores the following data:
\begin{itemize}
 \item \texttt{Id} --- The ghost's global tree index.
 \item \texttt{eclass} --- The type of the ghost tree.
 \item \texttt{tree\_to\_tree} --- An array giving for each face the global number
 of its face-neighbor.
 \item
   \texttt{tree\_to\_face} ---
   As above.
\end{itemize}

Since a ghost stores the global number of all of its face-neighbor trees,
we can locally compute all other processes that have this tree as a ghost by
combining the information from \texttt{O} and \texttt{tree\_to\_tree}.

\subsection{Updating local indices}

After partitioning, the local indices of the trees and ghosts change.
The new local indices of the local trees are determined by subtracting the global
index of the first local tree from the global index of each local tree.
The local indices of the ghosts are given by their position in the 
data array. 

Since the local indices change after repartitioning, we update 
the face-neighbor entries of the local trees to store those new values.
Because a neighbor of a tree can either be a local tree or a ghost on the 
previously owning process $\tilde p$ and become either a tree or a ghost on the new 
owning process $p$, there are four cases that we shall consider.

We handle these four cases in two phases, the first phase is carried out on
process $\tilde p$ before the tree is sent to $p$. In this phase we change
all neighbor entries of the trees that become local.
The second phase is done on $p$ after the tree was received from $\tilde p$.
Here we change all neighbor entries belonging to trees that become ghosts.

In the first phase,
$\tilde p$ has information about the first local tree on $\tilde p$ in the
old partition, its global number being $k_{\tilde p}$. Via
\texttt{O'} it also knows $k_p^\mathrm{new}$, the global index
of $p$'s first tree in the new partition.
Given a local tree on $\tilde p$ with local index $\tilde k$ in the old
partition we compute its new local index $k$ on $p$ as
\begin{equation}
  \label{eq:idupdate1}
  k = k_{\tilde p} + \tilde k - k_p^\mathrm{new},
\end{equation}
which is its global index minus the global index of the new first local tree.
Given a ghost $g$ on $\tilde p$ that will be a local tree on $p$, we
compute its local tree number as
\begin{equation}
  \label{eq:idupdate2}
 k = g.\mathrm{Id} - k_p^\mathrm{new}.
\end{equation}

In the second phase $p$, has received all its new trees and ghosts and
thus can give the new ghosts local indices to be stored in the \texttt{neighbors}
fields of the trees.
We do this by parsing, for each process $\tilde p \in R_p$ (in ascending order),
its ghosts and increment a counter.
For each ghost we parse its neighbors for local trees, and for any of these
we set the appropriate value in its \texttt{neighbors} field.

\subsection{\texttt{Partition\_cmesh} --- Algorithm
\ref{alg:partgiven}}

The input is a partitioned coarse mesh $C$ and a new partition layout \texttt{O'},
and the output is a new coarse mesh $C'$ that carries the same information as
$C$ and is partitioned according to \texttt{O'}.

This algorithm follows the method described in Section \ref{sec:partghosknow}
and is separated in two main phases, the \texttt{sending phase}
and the \texttt{receiving phase}.
In the former we iterate over each process $q\in S_p$ and decide which local trees
and ghosts we send to $q$.
Before sending, we carry out phase one of the update of the local tree numbers.
Subsequently, we receive all trees and ghosts from the processes in
$R_p$ and carry out phase two of the local index updating.

In the \texttt{sending phase} we iterate over the trees that we send to $q$.
For each of these trees we check for each neighbor (local tree and ghost)
whether we send it to $q$ as a ghost tree. This is the item 2 in the list 
of Section \ref{sec:partghosknow}.
The function \texttt{Parse\_neighbors} decides for a given local tree or ghost
neighbor whether it is sent to $q$ as a ghost.

\begin{algorithm}
\SetVlineSkip{1pt} 
\DontPrintSemicolon
\caption{\texttt{Partition\_cmesh}(\texttt{cmesh} $C$, \texttt{partition}
\texttt{O'})}  
\label{alg:partgiven}
  $p\gets$ this process 
  
  From \texttt{C.O} and \texttt{O'} determine $S_p$ and
  $R_p$. \Comment{See Section \ref{sec:detSpRp}}

  \tcc*[l]{Sending phase}
 
 \For {each $q\in S_p$}
 {
 $G\gets\emptyset$  \Comment{trees $p$ sends as ghosts to $q$}
  
  $s \gets$ first local tree to send to $q$.
  
  $e \gets$ last local tree to send to $q$.
  
  $T \gets \set{C.\texttt{trees}[s],\dots,C.\texttt{trees}[e]}$\Comment{local trees $p$ sends to $q$}

  \For {$k\in T$}
  {
  \texttt{Parse\_neighbors} ($C$, $k$, $p$, $q$, $G$, \texttt{O'})
  }
  \texttt{update\_tree\_ids\_phase1} $(T\cup G)$\Comment{See equations \eqref{eq:idupdate1} and \eqref{eq:idupdate2}}

  Send $T\cup G$ to process $q$
  }

  \tcc*[l]{Receiving phase}
 \For {each $q\in R_p$}
 {
   Receive $T[q] \cup G[p]$ from process $p$
 }
 
 $C'.\mathrm{trees}\gets \displaystyle\bigcup_{R_p} T[q]$
 \Comment{The updated tree array}
  
 $C'.\mathrm{ghosts}\gets \displaystyle\bigcup_{R_p} G[q]$
 \Comment{The updated ghost array}
  
 \texttt{update\_tree\_ids\_phase2} $(C.\mathrm{ghosts})$
 
 $C'.\texttt{O}\gets \texttt{O'}$

 \Return C'
   \vspace{1ex}
   \hrule
   \vspace{1ex}
 
 \tcc*[l]{decide which neighbors of $k$ to send as a ghost to $q$}
 \setcounter{AlgoLine}{0}
 \textbf{Function} \texttt{Parse\_neighbors}(\texttt{cmesh} $C$, \texttt{tree} $k$,
 \texttt{processes} $p$, $q$, \texttt{Ghosts} $\ast G$, \texttt{partition} \texttt{O'})

   \For {$u\in k.\texttt{tree\_to\_tree}\ohne\set{C.\texttt{trees}[s],\dots,C.\texttt{trees}[e]}$
   }
   {
    \eIf {$0\leq u<k_{p}$
    $\algoand \texttt{Send\_ghost}(C,\texttt{ghost}(u),q,\texttt{O'})$}
    {
    \If {$u+k_p\notin f'(q)$}
        {
          $G\gets G\cup \set{\texttt{ghost}(u)}$ \Comment{local tree $u$ gets ghost of $q$}
        }
    }
    (\IfComment{$k_{p}\leq u$})
    {
     $g\gets C.\mathrm{ghosts}[u-n_p]$
     \If {$g.Id\notin f(q)$ \algoand \texttt{Send\_ghost}$(C,g,q,\texttt{O'})$}
     {       
	  $G\gets G\cup \set{g}$ 
	  \Comment{$g$ is a ghost of $q$}
     }
    }   
   }
   \vspace{1ex}
   \hrule
   \vspace{1ex}
  \tcc*[l]{Subroutine to decide whether to send a ghost or not}
 \setcounter{AlgoLine}{0}
 \textbf{Function} \texttt{Send\_ghost}(\texttt{cmesh} $C$, \texttt{ghost} $g$, \texttt{process} $q$,
 \texttt{partition} \texttt{O'})
 $S\gets\emptyset$\\
 \For {$u\in g.\texttt{tree\_to\_tree}$} 
 {    
   \For {$q'$ with $u$ is a local tree of $q'$}
   {
   \If {$q'$ sends $u$ to $q$}
    {
      $S\gets S\cup \set{q'}$ 
    }
  }
 }
 \eIf {$q\notin S$ \algoand $p = \min S$} 
 {
  \Return True \Comment{$p$ is the smallest rank sending trees to $q$}
 }
 {
  \Return False  
 }
\end{algorithm}

\section{Numerical Results}

The run time results that we present here have been obtained using the Juqueen
supercomputer at Forschungszentrum J{\"u}lich, Germany.
It is an IBM BlueGene/Q system with 28,672 nodes consisting of IBM PowerPC
A2 processors at 1.6~GHZ with 16~GB RAM per node \cite{Juqueen}.
Each compute node has 16 cores and is capable of running up to 64 MPI
processes using multithreading.

\subsection{How to obtain example meshes}

To measure the performance and mem\-ory consumption of the algorithms presented
above, we would like to test the algorithms on coarse meshes
that are too big to fit into the memory of a single process,
which is 1 GB on Juqueen if we use 16 MPI ranks per node.
We consider three approaches to do construct such meshes:
\begin{enumerate}
\item Use an external parallel mesh generator.
\item Use a serial mesh generator on a large-memory machine, transfer the
  coarse mesh to the parallel machine's file system, and read it using
  (parallel) file I/O.
\item
  Create a big coarse mesh by forming the disjoint union of smaller coarse
  meshes over the individual processes.
\end{enumerate}

Due to a lack of availability of open source parallel mesh
generators
we restrict ourselves to the second and third method.
These have the advantage that we can start with initial coarse meshes that fit
into a single process' memory, such that we can work with serial mesh
generating software.
In particular we use \texttt{gmsh}, \texttt{Tetgen} and \texttt{Triangle}
for simplicial meshes \cite{GeuzaineRemacle09, Shewchuk96, Si06}.

The third method is especially well suited for weak scaling studies.
The small coarse meshes can be created programmatically, communication-free on
each process, which we exploit below for tests with hexahedral trees.
They may be of same or different sizes between the processes.
In the former approach,
the number of local trees is, by definition, the same on each process.

We discuss two examples below: one examining purely the coarse mesh
partitioning without regard for a forest and its elements, and another in which
we drive the coarse mesh partitioning by a dynamically changing forest of
elements.
The latter manages shared trees and thus fully executes the algorithmic ideas
put forward above.

\subsection{Disjoint bricks}

In our first example we conduct a strong and weak scaling study of coarse mesh
partitioning and test the maximal number of cubical trees that we can support
before running out of memory.
To obtain reliable weak scaling results, we keep the same per-process
number of trees while increasing the total number of processes.
We achieve this by constructing
an $n_x \times n_y \times n_z$ brick of cubical trees on each process,
using three constant parameters $n_x, n_y$ and $n_z$.
We repartition this coarse mesh once, by the rule that each rank $p$ sends
43\% of its local trees to the rank $p+1$ (except the biggest rank $P-1$, which
keeps all its local trees).
We choose this odd percentage to create non-trivial boundaries between the
regions of trees to keep and trees to send.
See Figure \ref{fig:cmesh_partition1} for a depiction of the partitioned coarse
mesh on six processes.
The local bricks are created locally as \pforest connectivities with
\texttt{p4est\_connectivity\_new\_brick} and are then reinterpreted in parallel
as a distributed coarse mesh data structure.

\begin{figure}
\includegraphics[width=\textwidth]{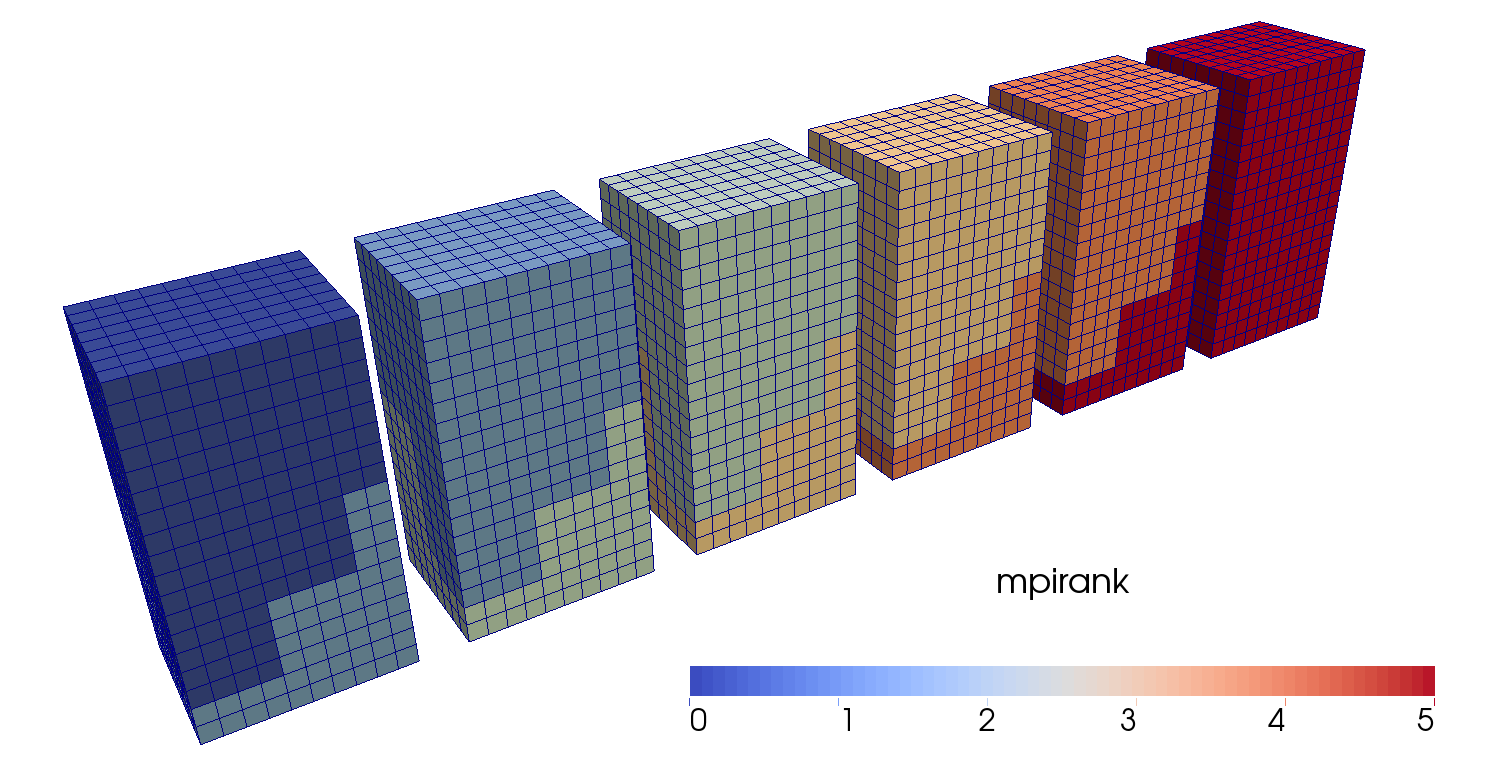}
\caption{Depicting the structure of the coarse
mesh for an example with six processes
that we use to measure the maximum possible mesh sizes and scalability.
Before partitioning, the coarse mesh local to each process is created as one
$n_x\times n_y \times n_z$ block of cubical trees.
We repartition the mesh such that each process sends 43\% of its local trees to
the next process.
The picture shows the resulting partitioned coarse mesh with parameters $n_x =
10, n_y = 18$ and $n_z = 8$ and color coded MPI rank.}
\label{fig:cmesh_partition1}
\end{figure}

\begin{figure}
\center
\begin{minipage}{0.49\textwidth}
\includegraphics[width=\textwidth]{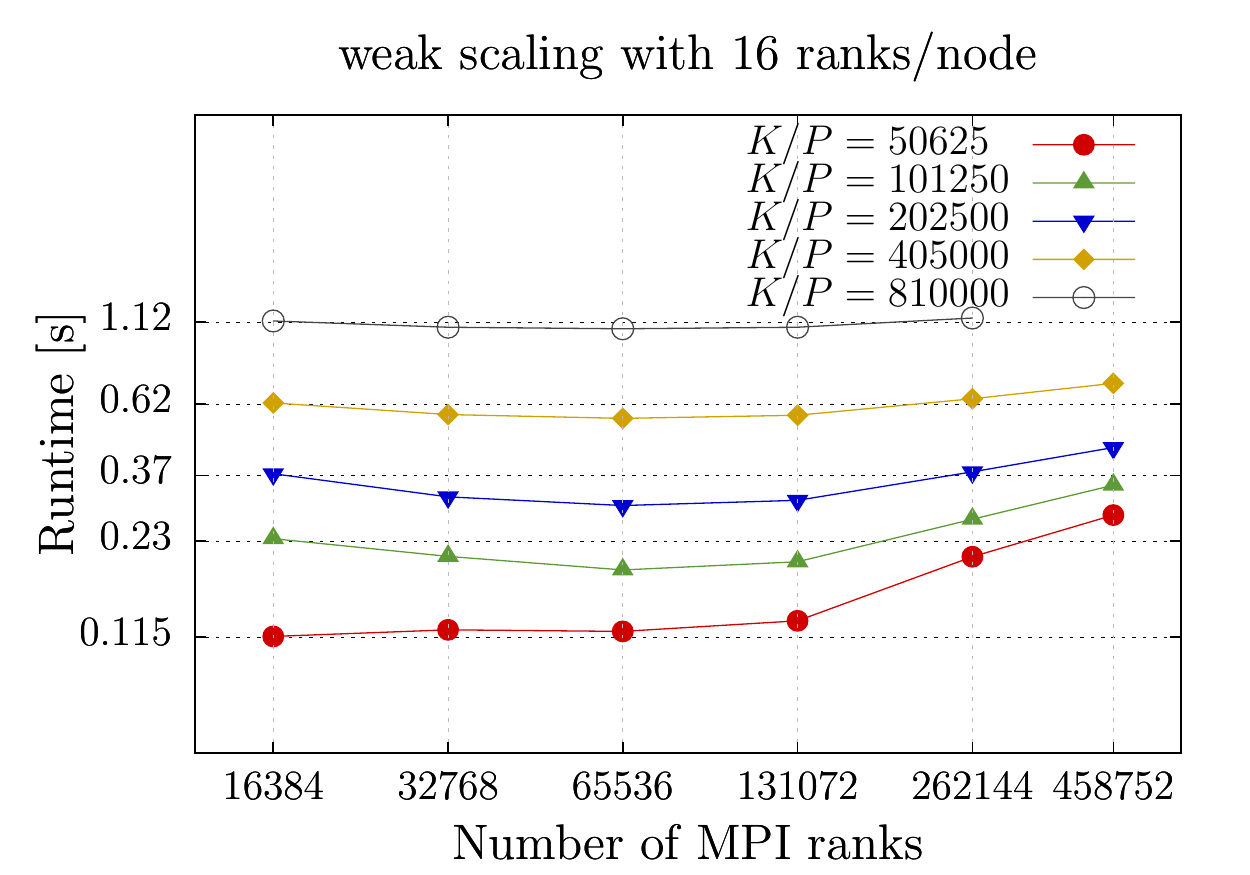}
\end{minipage}
\begin{minipage}{0.49\textwidth}
\includegraphics[width=\textwidth]{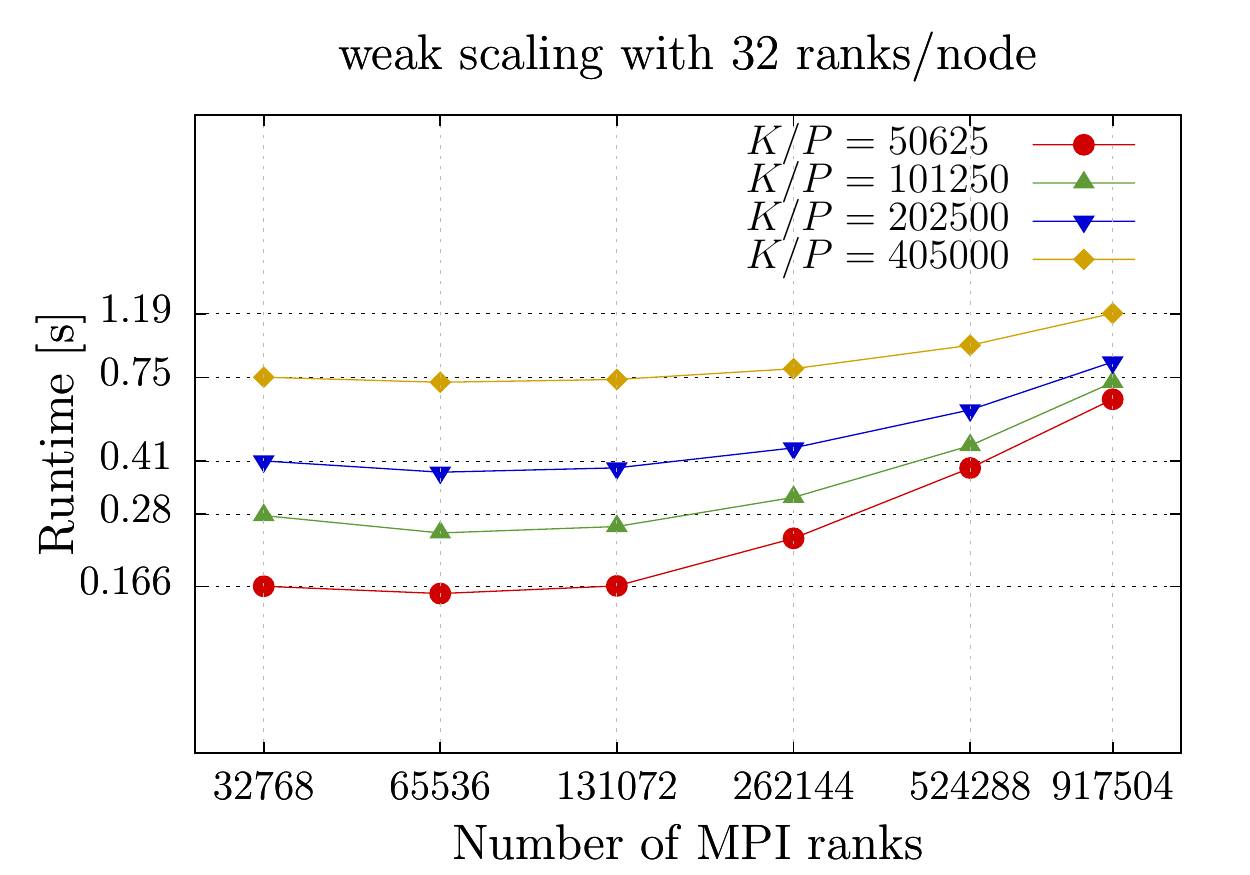}
\end{minipage}
\caption{Weak scaling of \texttt{partition} with disjoint bricks.
Left: 16 ranks per node. Right: 32 ranks per node.
We show the run times for the baseline on the y-axis.
On the left hand side the run time for the biggest 262,144-process run is
1.15 seconds, and
the run time for the biggest 458,752 run is 0.719 seconds. We obtain
an efficiency of 97.4\%, resp.\ 86.2\% compared to the baseline
of 16,384 MPI ranks.}
\label{fig:cmesh_disjoint_scaling}
\end{figure}

\begin{figure}
\center
\begin{minipage}{0.49\textwidth}
\includegraphics[width=\textwidth]{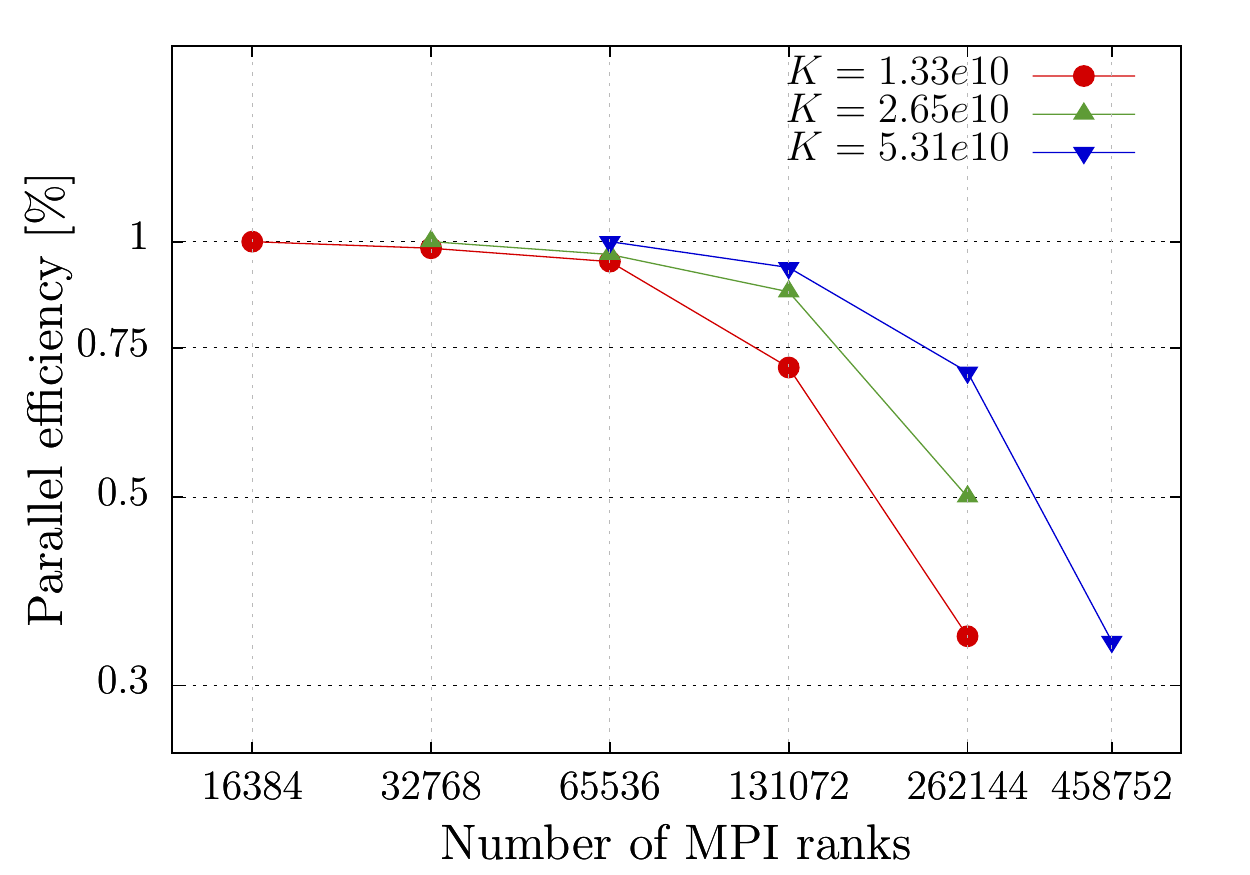}
\end{minipage}
\begin{minipage}{0.49\textwidth}
\includegraphics[width=\textwidth]{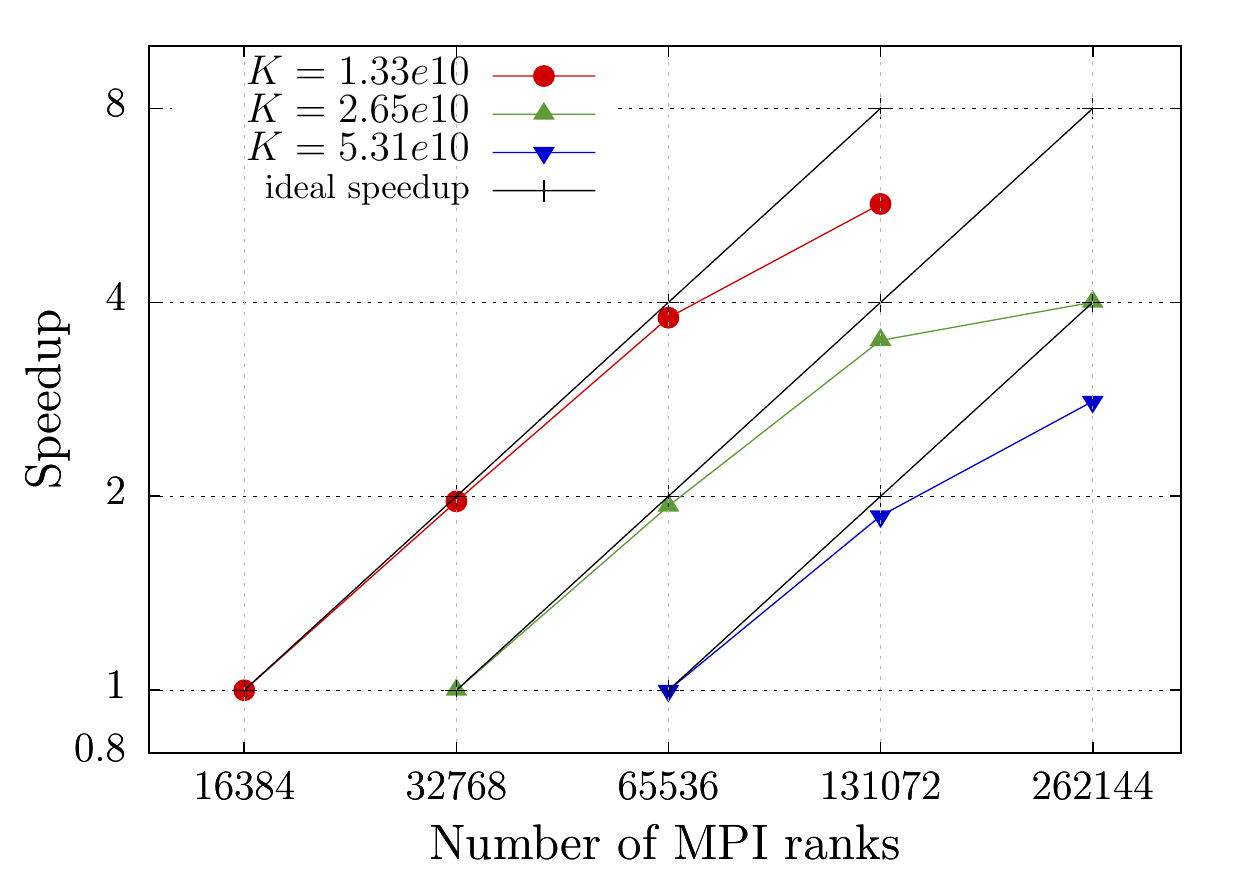}
\end{minipage}
\caption{Strong scaling of \texttt{Partition\_cmesh} for the disjoint bricks
example on Juqueen with 16 ranks per compute node,
for three runs with 1.33e10, 2.65e10, and 5.31e10 coarse mesh elements.
We show the parallel efficiency on the left and the speedup on the right.
The absolute run times for 262,144 processes are 0.206, 0.270 and 0.380 seconds.}
\label{fig:cmesh_disjoint_scaling_strong}
\end{figure}

\begin{table}
\center
\begin{tabular}{|rrrrr|}
 \multicolumn{5}{c}{Run time tests for \texttt{Partition\_cmesh}}\\[2ex] \hline
 \multicolumn{5}{|c|}{\mytabvspace 131072 MPI ranks (16 ranks per node)}\\ \hline
 \mytabvspace mesh size & per rank & trees (ghosts) sent & Time [s] & Factor\\ \hline
 \mytabvspace   6,635e9\phantom{0} & 50,625  & 21,768 (3,414) &  0.13 & --\\
               1,327e10            & 101,250 & 43,537 (5,504) &  0.20 & 1.53\\
               2,654e10            & 202,500 & 87,074 (6,607) &  0.31  & 1.56\\
               5,308e10            & 405,000 & 174,149  (11,381) & 0.57 & 1.85\\
               1,062e11            & 810,000 & 348,297  (22,335) & 1.08  & 1.89\\
 \hline\hline
 \multicolumn{5}{|c|}{917504 MPI ranks (32 ranks per node)}\\ \hline
 \mytabvspace mesh size & per rank & trees (ghosts) sent & Time [s] & Factor \\ \hline
 \mytabvspace   4.645e10 & 50,625   & 21,768 (3,413) & 0.64      & --\\
 \mytabvspace   9.290e10 & 101,250  & 43,537 (5,504) & 0.72   & 1.13 \\
 \mytabvspace   1.858e11 & 202,500  & 87,075 (6,607) & 0.84 & 1.12 \\
 \mytabvspace   3,716e11 & 405.000  & 174,150 (11,383) & 1.19    & 1.42 \\ \hline
\end{tabular}
\caption{The run times of \texttt{Partition\_cmesh} for 131,072 processes with
16 processes per node (top) and 917,504 processes with 32 processes per node
(bottom). The biggest coarse mesh that we created during the tests has 371
Billion trees. The last column is the quotient of the current run time divided
by the previous run time. Since we double the mesh size in each step, we would
expect an increase in run time of a factor of 2, which hints at parallel
overhead becoming negligible in the limit of many elements per process.%
}
\label{tab:131k917k}
\end{table}

We perform strong and weak scaling studies on up to 917,504 MPI ranks and display
our results in Figures~\ref{fig:cmesh_disjoint_scaling} and
\ref{fig:cmesh_disjoint_scaling_strong} and Table~\ref{tab:131k917k}.
We show the results of one study with 16 MPI ranks per compute node, thus 1GB
available memory per process, and one with 32 MPI ranks per compute node,
leaving half of the memory per process.
In both cases we measure run times for different mesh sizes per process.
We observe that even for the biggest meshes of 405k resp.\ 810k coarse mesh elements
per process the absolute run times of partition are below 1.2 seconds.
Furthermore, we measure a weak scaling efficiency of 97.4\% for the 810k
mesh on 262,144 processes and 86.2\% for the 405k
mesh on 458,752 processes.
The biggest mesh that we created is partitioned between 917,504 processes and
uses 405k mesh elements per process for a total of over 371e9 coarse mesh
elements.

Additionally, in Table \ref{tab:smallmeshes} we show run times for
\texttt{Partition\_cmesh} with small coarse meshes, where the number of
elements is roughly on the order of the number of processes.
These tests show that
for such small meshes the run times are in the order of milliseconds.
Hence, for small meshes there is no disadvantage in using a partitioned coarse
mesh over a replicated one, when each process holds a full copy.

\begin{table}
\center
\begin{tabular}{|r|r|r|}\hline
\# MPI ranks  & mesh elements & run time [s] \\\hline
1024 &   4096 & 0.00136\\ 
1024 &   8192 & 0.00149\\
1024 &  16384 & 0.00142\\
  64 &    105 & 0.00122\\ 
  32 &    105 & 0.00789 \\\hline

64 & 3200  &  0.000293\\
64 & 19200 &  0.000865\\\hline
\end{tabular}
\caption{Run times for \texttt{Partition\_cmesh} for relatively small coarse
meshes. The bottom part of the table was not computed on Juqueen but on a
local institute cluster of 78 nodes with 8 Intel Xeon CPU E5-2650 v2 @ 2.60GHz
each.}
\label{tab:smallmeshes}
\end{table}

\subsection{An example with a forest}
In this example we partition a tetrahedral coarse mesh according to a parallel
forest of fine elements.
Opposed to the previous test, where we exploited the
maximum possible coarse element numbers, we now test mesh sizes that can occur in
realistic use cases.

When simulating shock waves or two-phase flows, there is often an interface along
which a finer mesh resolution is desired in order to minimize computational
errors.
Motivated by this example, we create the forest mesh as an initial uniform
refinement of the coarse mesh with a specified level $\ell$ and refine it in a
band along an interface defined by a plane in $\IR^3$ up to a maximum
refinement level $\ell+k$.
As refinement rule we use 1:8 red refinement \cite{Bey92} together with the
tetrahedral Morton space-filling curve \cite{BursteddeHolke16}. We move the 
interface through the domain with a constant velocity. Thus, in each time
step the mesh is refined and coarsened, and therefore we
repartition it to maintain an optimal load-balance. We measure run times for
both coarse mesh and forest mesh partitioning for three time steps.

Our coarse mesh consists of tetrahedral trees modelling a brick with spherical
holes in it. To be more precise, the brick is built out of $n_x\times n_y
\times n_z$ tetrahedralized unit cubes and each of those has one spherical hole
in it; see Figures \ref{fig:cmesh_partition2} and \ref{fig:forest_partition2}
for a small example mesh.

We create the mesh in serial using the generator \texttt{gmsh}
\cite{GeuzaineRemacle09}.
We read the whole file on a single process, and thus
use a local machine with 1 terabyte memory for preprocessing.
On this machine we partition the
coarse mesh to several hundred processes and write one file for each partition.
This data is then transferred to the supercomputer.
The actual computation consists of reading the coarse mesh from files, creating
the forest on it, and partitioning the forest and the coarse mesh
simultaneously.
To optimize memory while loading the coarse mesh, we open at most one
partition file per compute node.

The coarse mesh that we use in these tests has parameters $n_x = 26, n_y = 24,
n_z = 18$, thus 11,232 unit cubes. Each cube is tetrahedralized with
about 34,150 tetrahedra, and the whole mesh consists of 383,559,464 trees.
In the first test, we create a forest of uniform level 1 and
maximal refinement level 2, and in the second a forest of uniform level 2
and maximal refinement level 3. The forest mesh in the first test consists of
approximately 2.6e9 elements. In the second test we also use a broader
band and obtain a forest mesh of 25e9 tetrahedra.

We show the run time results and further statistics for coarse mesh partitioning
in Table \ref{tab:brickexample_cmesh} and results
for forest partitioning in Table \ref{tab:brickexample_forest}.
We observe that the run time for \texttt{Partition\_cmesh} is between $0.10$ and
$0.13$ seconds, and that about 88\%, respectively 98\%, of all processes share
local trees with other processes.
The run times for forest partition are below $0.22$ seconds for the first example
and below $0.65$ seconds for the second example.
Thus, the overall run time for a complete mesh partition is below $0.5$ seconds
for the first and below $1$ second for the second example.

We also run a third test on 458,752 MPI ranks and refine the forest to a maximum
level of four. Here, the forest mesh has 167e9 tetrahedra, which we partition
in under 0.6 seconds. The coarse mesh partition routine runs in about 0.2
seconds. Approximately 60\% of the processes have a shared tree in the coarse
mesh. We show the results from this test in Table 
\ref{tab:brickexample_forest_458k}.

\begin{figure}
\center
\includegraphics[width=0.7\textwidth]{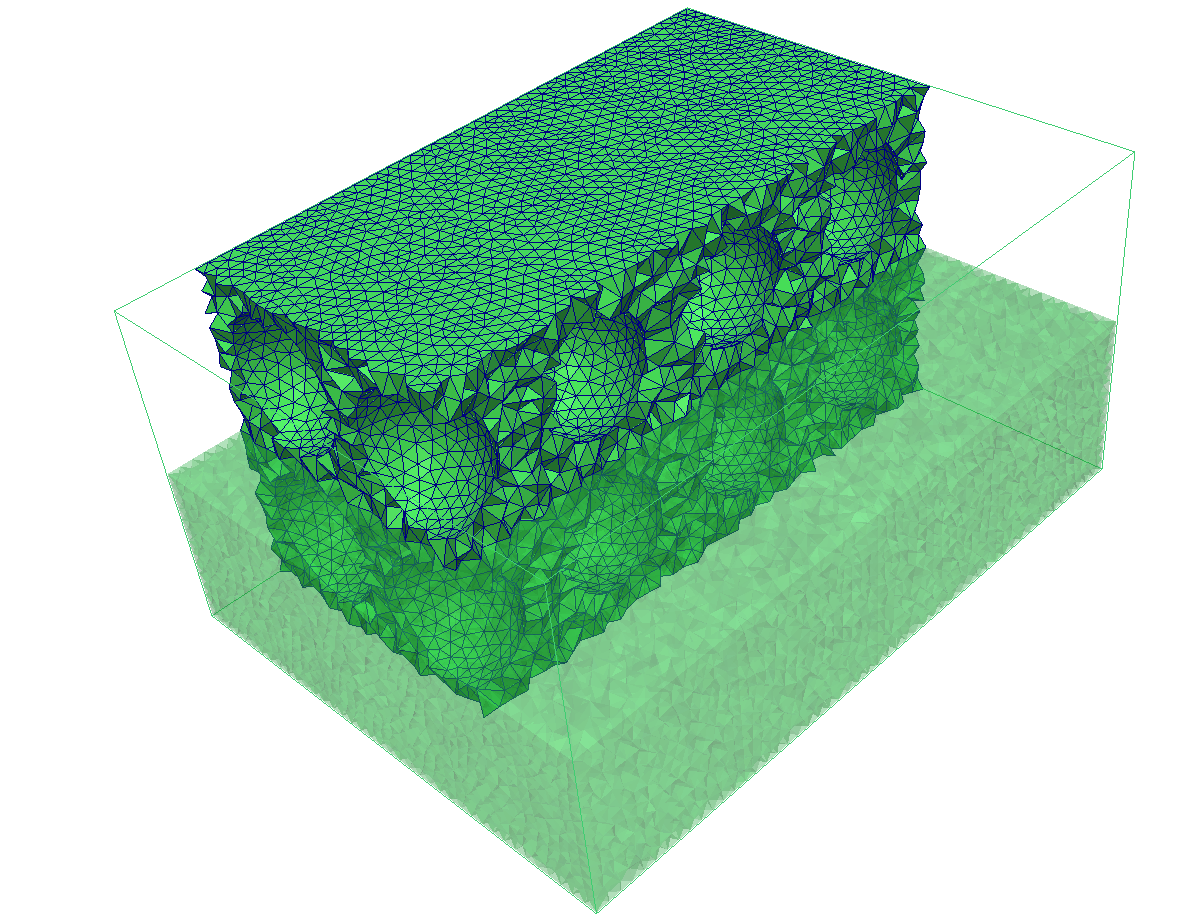}
\caption{The coarse mesh connectivity that we use for the
  partition tests motivated by an adapted forest.
It consists of $n_x\times n_y \times n_z$ cubes with one spherical hole each.
For this picture we use $n_x = 4,\, n_y = 3,\, n_z = 2$, and
each cube is triangulated with approximately 7575 tetrahedra.
For illustration purposes we show parts of the mesh opaque and other parts
are invisible.}
\label{fig:cmesh_partition2}
\end{figure}

\begin{figure}
\center
\includegraphics[width=0.3\textwidth]{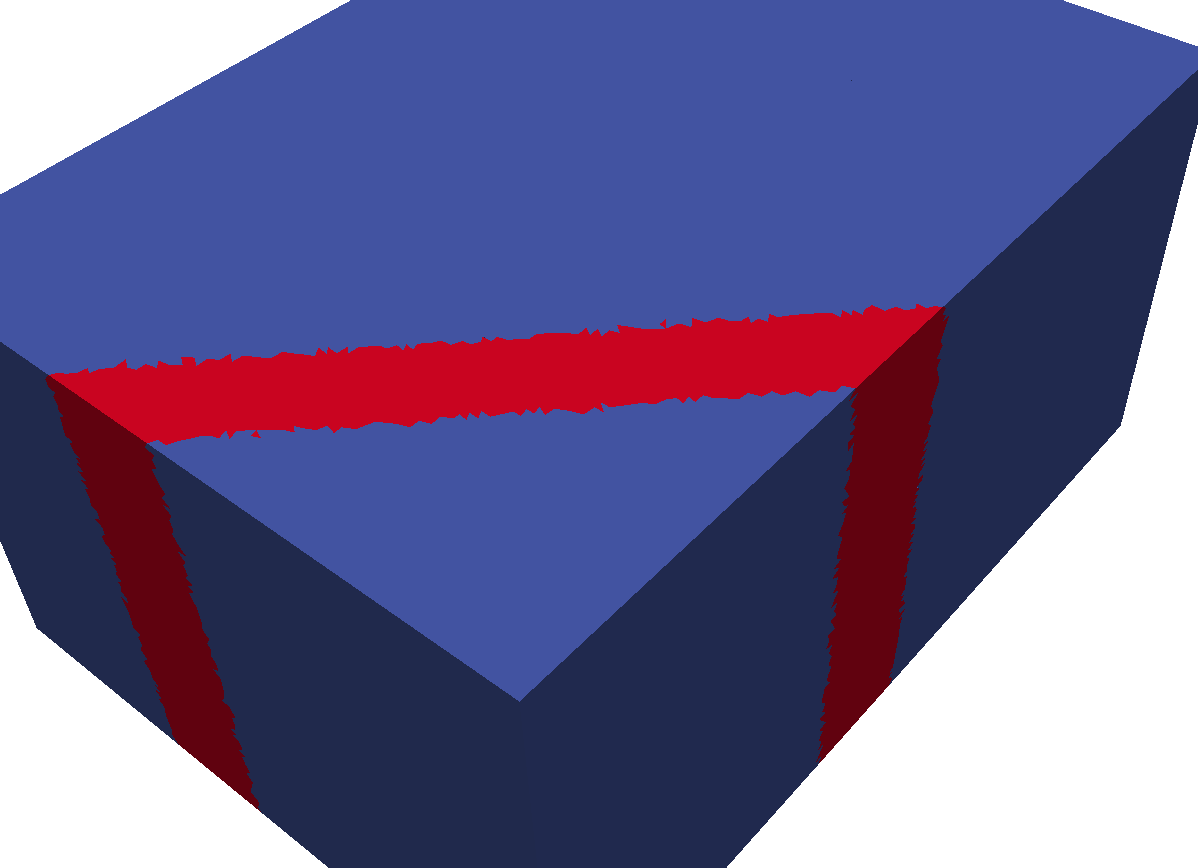}
\includegraphics[width=0.3\textwidth]{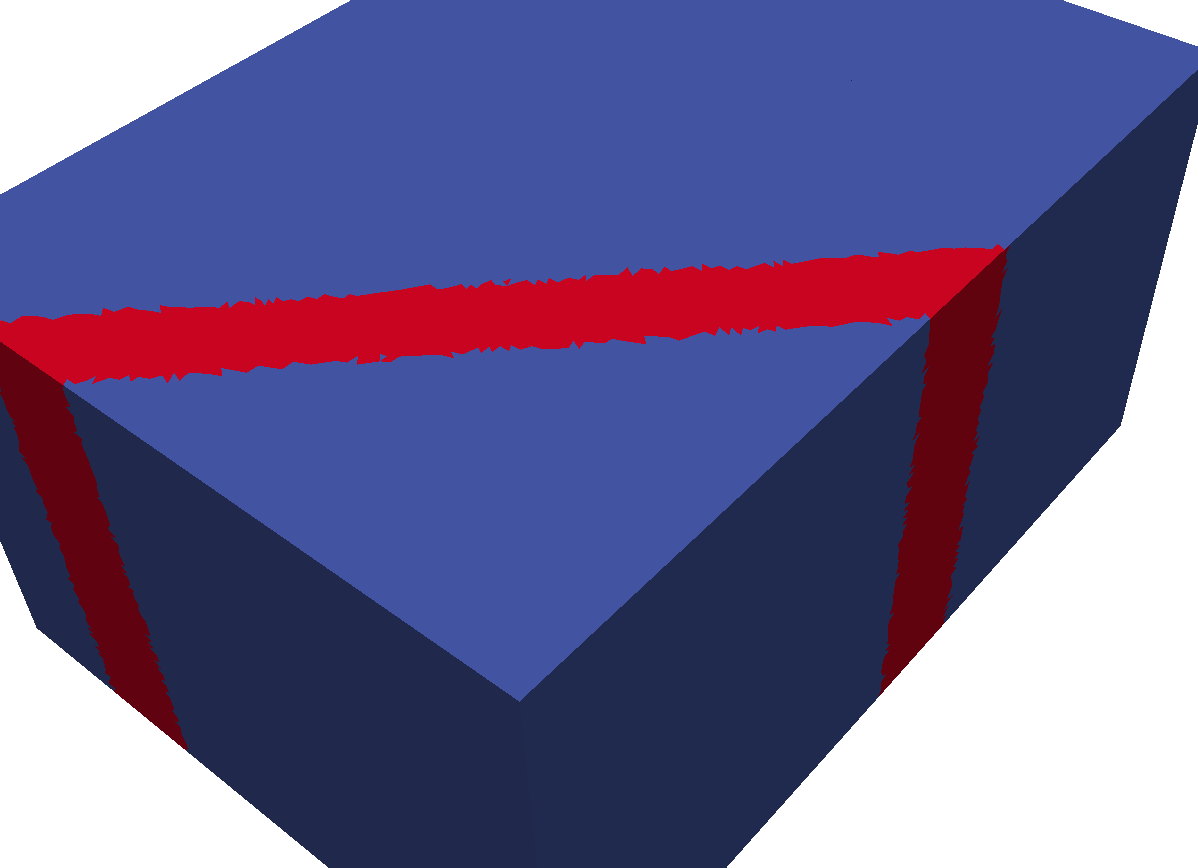}
\includegraphics[width=0.3\textwidth]{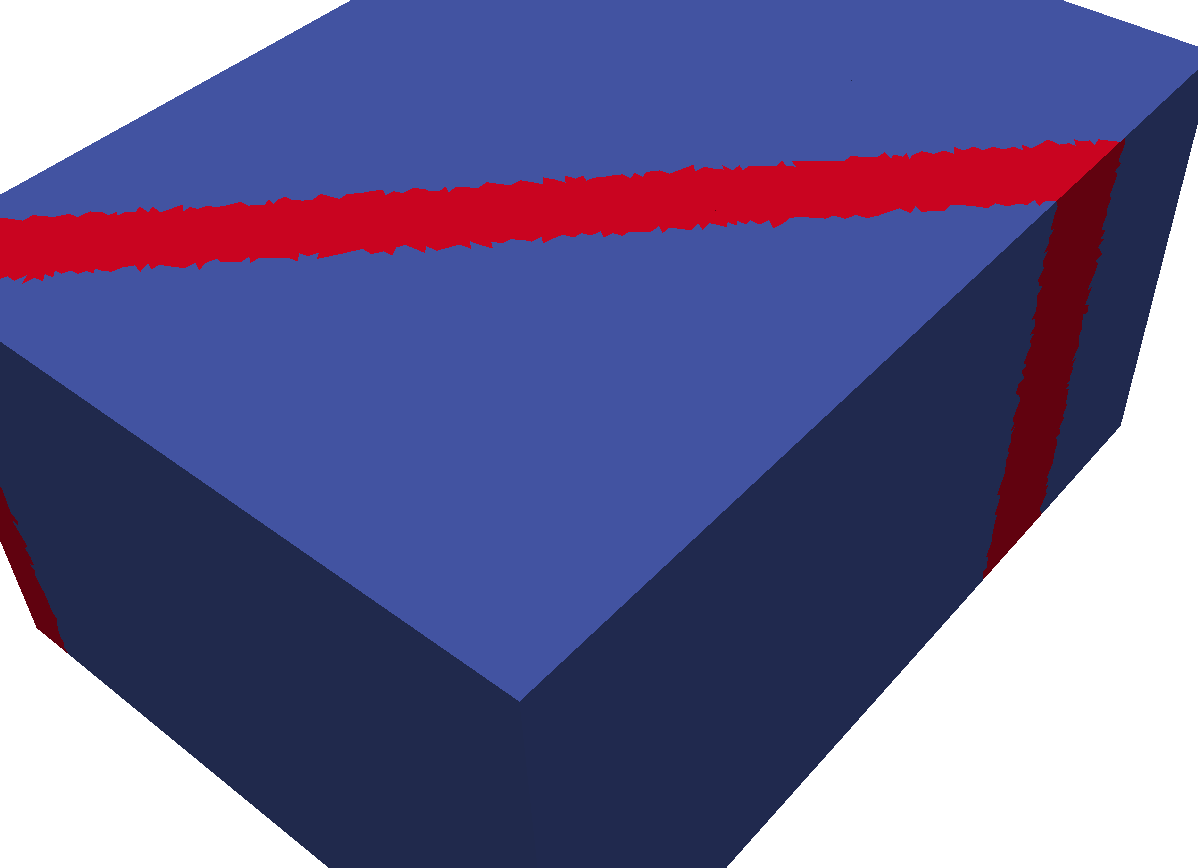}
\caption{An illustration of the band of finer forest mesh elements in the
example. The region of finer mesh elements moves through the mesh in each time
step. From left to right we see $t=1$, $t=2$, and $t=3$. In this illustration,
elements of refinement level 1 are blue and elements of refinement level 2 are
red.}
\label{fig:forest_partition2}
\end{figure}

\begin{table}
\center
\caption*{Coarse mesh partition on 8,192 MPI ranks}
\begin{tabular}{|c|r|r|r|r|r|}
\hline  
 t & trees (ghosts) sent & data sent [MiB] & $|S_p|$ & shared trees & run time [s]\\ \hline
1 & 25,116.7 (11,947.5) & 4.945 & 2.274 & 7178 & 0.103 \\
2 & 34,860.3 (16,854.9) & 6.880 & 2.753 & 7176 & 0.110 \\
3 & 36,386.0 (17,567.6) & 7.180 & 2.823 & 7182 & 0.112 \\\hline
1 & 39,026.2 (18,333.9) & 7.670 & 2.962 & 8096 & 0.128 \\
2 & 38,990.1 (18,268.0) & 7.660 & 2.946 & 8085 & 0.129 \\
3 & 38,941.7 (18,073.9) & 7.639 & 2.933 & 8085 & 0.128 \\
\hline
\end{tabular}
\caption{\texttt{Partition\_cmesh} with a coarse mesh of 324,766,336
tetrahedral trees on 8192 MPI ranks. We measure mesh repartition for three
time steps. For each we show process average values of the
number of trees ghosts, and total number of bytes that each process sends to
other processes. The average number of other processes to which a process sends
($|S_p|$) is in each test below three. We also give the total number of shared
trees in the mesh, 8191 being the maximum possible value. We round all numbers
to three significant figures. $1 \textrm{ MiB} = 1024^2\textrm{ bytes}.$}
\label{tab:brickexample_cmesh}
\end{table}

\begin{table}
\center

\caption*{Forest mesh partition on 8,192 MPI ranks}
\begin{tabular}{|c|r|r|r|r|}
\hline  
t & mesh size & elements sent & data sent [MiB]  & run time [s] \\ \hline
1 &  2,622,283,453 & 203,858 & 3.494      & 0.215\\  
2 &  2,623,842,241 & 281,254 & 4.824      & 0.215\\  
3 &  2,626,216,984 & 293,387 & 5.032      & 0.214\\\hline  
1 & 25,155,319,545 & 3,013,230 & 46.574 & 0.642\\ 
2 & 25,285,522,233 & 3,008,800 & 46.506 & 0.640\\ 
3 & 25,426,331,342 & 2,991,990 & 46.248 & 0.645\\ 

\hline
\end{tabular}
\caption{For the same example as in Table \ref{tab:brickexample_cmesh} we display
statistics and run times for the forest mesh partition. We show the total number
of tetrahedral elements, and the average count of elements and bytes that each
process sends to other processes (their count is the same as in Table
\ref{tab:brickexample_cmesh}). We round all numbers to three significant figures.}
\label{tab:brickexample_forest}
\end{table}

\begin{table}
\center
\caption*{Coarse mesh partition on 458,752 MPI ranks}
\begin{tabular}{|c|r|r|r|r|r|}\hline
t & trees (ghosts) sent & data sent [MiB] & $|S_p|$ & shared trees & run time [s]\\ \hline
1 & 703.976 (2444.11) & 0.267 & 2.986 & 280,339 & 0.207 \\
2 & 707.368 (2455.95) & 0.269 & 2.996 & 281,694 & 0.204 \\
3 & 707.904 (2457.70) & 0.269 & 2.997 & 281,900 & 0.204 \\
\hline
\end{tabular}

\caption*{Forest mesh partition on 458,752 MPI ranks}
\begin{tabular}{|c|r|r|r|r|}
\hline  
t & mesh size        & elements sent & data sent [MiB]  & run time [s] \\ \hline
1 & 167,625,595,829  & 362,863 & 5.548 & 0.522 \\  
2 & 167,709,936,554  & 364,778 & 5.577 & 0.578 \\  
3 & 167,841,392,949  & 365,322 & 5.585 & 0.567 \\\hline  
\end{tabular}
\caption{Run times for coarse mesh and forest partition for the brick with holes
on 458,752 MPI ranks. The setting and the coarse mesh are the same as in 
Table \ref{tab:brickexample_cmesh} despite that for the forest we use a initial
uniform level three refinement with a maximum level of four.}
\label{tab:brickexample_forest_458k}
\end{table}

\section{Conclusion}

In this manuscript we propose an algorithm that executes dynamic and in-core
coarse mesh partitioning.
In the context of forest-of-trees adaptive mesh refinement, the coarse mesh
defines the connectivity of tree roots, which is used in all neighbor query
operations between elements.
This development is motivated by simulation problems on complex domains that
require large input meshes.
Without partitioning, we will run out of memory around one million coarse
elements, and with static or out-of-core partitioning, we might not
have the flexibility to transfer the tree meta data as required by the change
in process ownership of the trees' elements, which occurs in every AMR cycle.
With the approach presented here, this can be performed with run times that are
significantly smaller than those for partitioning the elements, even
considering that SFC methods for the latter are exceptionally fast in absolute
terms.
Thus, we add little to the run time of all AMR operations combined.

Our algorithm guarantees both, that the number of fine mesh elements is
distributed equally among the processes using a space-filling curve, and that
each process can provide the necessary connectivity data for each of its fine
mesh elements.
We handle the communication without handshaking and develop a communication
pattern that minimizes data movement.
This pattern is calculated by each process individually, reusing information
that is already present.

Our implementation scales up to 917e3 MPI processes and up to 810e3 coarse mesh
elements per process, where the largest test case consists of .37e12 coarse
mesh elements.
What remains to be done is extending the partitioning of ghost trees to edge
and corner neighbors, since only face neighbor ghost trees are presently
handled.
It appears that the structure of the algorithm will allow this with little
modification.

\section*{Acknowledgments}

The authors gratefully acknowledge travel support by the Bonn Hausdorff Center
for Mathematics (HCM) funded by the Deut\-sche For\-schungs\-ge\-mein\-schaft
(DFG).
H.\ acknowledges additional support by the Bonn International Graduate School
for Mathematics (BIGS) as part of HCM.
We use the interface to the MPI3 shared array functionality written by Tobin
Isaac, University of Chicago, to be found in the files \texttt{sc\_shmem} of
the \texttt{sc} library.

The authors would like to thank the Gauss Centre for Supercomputing (GCS) for
providing computing time through the John von Neumann Institute for Computing
(NIC) on the GCS share of the supercomputer JUQUEEN
at J{\"u}lich Supercomputing Centre (JSC).
GCS is the alliance of the three
national supercomputing centres HLRS (Universit{\"a}t Stuttgart), JSC
(Forschungszentrum J{\"u}lich), and LRZ (Bayerische Akademie der
Wissenschaften), funded by the German Federal Ministry of Education and
Research (BMBF) and the German State Ministries for Research of
Baden-W{\"u}rttemberg (MWK), Bayern (StMWFK) and Nordrhein-Westfalen (MIWF).

\bibliographystyle{plain}
\bibliography{coarse_mesh.bbl}
\end{document}